\documentclass[11pt]{article}

\usepackage[margin=1in]{geometry}
\usepackage{amsmath,amssymb,amsthm,amsfonts,hyperref,mathtools,tikz}
\usepackage{bbm}
\usepackage{enumitem}
\usepackage{microtype}
\usepackage{tikz-cd}
\usepackage[noabbrev,capitalise,nameinlink]{cleveref}
\usepackage{thm-restate}
\usepackage{stmaryrd}
\usepackage{authblk}

\newcommand{\N}{\mathbb N}

\newcommand{\unif}[1]{\mathsf H_{#1}}

\DeclareMathOperator{\diam}{diam}

\newtheorem{theorem}{Theorem}
\newtheorem{conjecture}[theorem]{Conjecture}

\newtheorem{lemma}[theorem]{Lemma}
\newtheorem{corollary}[theorem]{Corollary}

\theoremstyle{remark}
\newtheorem{remark}[theorem]{Remark}
\newtheorem{observation}[theorem]{Observation}

\crefname{observation}{Observation}{Observations}
\crefname{item}{Item}{Items}

\title{Ramsey Obstructions to Disambiguation}

\author[1]{Romain Bourneuf}
\affil[1]{Univ. Bordeaux, CNRS, Bordeaux INP, LaBRI, UMR 5800, F-33400 Talence, France.}
\author[2]{Antonin Kiladjian}
\affil[2]{Univ. Lyon, ENS de Lyon, UCBL, CNRS, LIP, France}
\author[2]{Stéphan Thomassé}
\date{}

\begin{document}
\maketitle

\begin{abstract}

A \emph{partial matrix} has entries in $\{0,1,\star\}$, and a \emph{disambiguation} replaces each $\star$ by $0$ or $1$. We construct partial matrices whose fully specified submatrices satisfy strong restrictions, yet every disambiguation contains every binary matrix of a prescribed size.

Our first result answers a question of Alon, Hanneke, Holzman and Moran on the disambiguation of linear classifiers with margin. 
For $0<\varepsilon<\pi/2$, let $M_\varepsilon^d$ be the partial matrix indexed by points of the unit sphere $\mathbb S^d$, with entry $0$ for pairs at spherical distance at most $\varepsilon$, $1$ for pairs at distance at least $\pi-\varepsilon$, and $\star$ otherwise. Although these matrices have VC-dimension bounded independently of $d$, we prove that every disambiguation contains every binary $k\times k$ matrix once $d$ is sufficiently large.
This also yields a partial concept class of Littlestone dimension $1$ with no disambiguation of finite VC-dimension.

We also construct, for every $k$, a finite partial matrix whose fully specified $2\times2$ submatrices are all constant, while every disambiguation contains every binary $k\times k$ matrix. A symmetric analogue holds for partial graphs: for every $k$, there exists a partial graph of VC-dimension at most $1$ whose fully specified induced subgraphs are all cliques or stable sets, yet every disambiguation contains every $k$-vertex graph as an induced subgraph.

A disambiguation can be viewed as a $2$-coloring of the unspecified entries, making Ramsey theory a natural framework for forcing prescribed patterns. Our proofs draw on two recent Ramsey theorems: the geometric argument uses Pálvölgyi's Dense Block theorem, while the combinatorial constructions rely on the girth Ramsey theorem of Reiher and Rödl, a suitable strengthening of the induced Ramsey theorem.

\end{abstract}

\section{Introduction}

A \emph{partial (binary) matrix} is a matrix with entries in $\{0,1,\star\}$, where $\star$ denotes an unspecified entry.
A \emph{disambiguation} is obtained by replacing each $\star$ entry independently by either $0$ or $1$.
A natural question is whether a partial matrix can be substantially simpler than each of its disambiguations.
In other words, can filling the unspecified entries inevitably create complexity that is absent from the partial matrix?

This question arises naturally in learning theory.
Viewing rows as concepts and columns as domain points, a partial matrix represents a class of partially defined binary classifiers.
%Learning is required only from examples on which the target concept is defined.
Such classes naturally express margin assumptions: a separating hyperplane need only label points sufficiently far from its boundary.
Partial classes were studied by Long~\cite{L01} in agnostic learning, and Alon, Hanneke, Holzman, and Moran~\cite{AHHM22} developed their PAC learning theory.

The relevant measure of complexity is VC-dimension: the largest number of columns on which the rows realize every binary pattern using only specified entries.
Equivalently, the VC-dimension of a partial matrix is the supremum VC-dimension of a pure\footnote{A submatrix of a partial matrix $M$ is \emph{pure} if all its entries are specified.} submatrix.
Alon, Hanneke, Holzman, and Moran~\cite{AHHM22} showed that finite VC-dimension characterizes PAC learnability of partial concept classes, just as it does for total classes.
However, they also constructed a partial class of VC-dimension $1$ whose every disambiguation has infinite VC-dimension, answering a question originating in the work of Attias, Kontorovich, and Mansour on adversarially robust learning~\cite{AKM22}.
Thus leaving entries unspecified can be essential to low complexity: the learnability of a partial class need not be explained by any learnable class of total extensions.

\smallskip

Their result leaves open which additional restrictions on a partial matrix can guarantee a disambiguation of bounded VC-dimension.
We show that the obstruction persists both for natural geometric examples and under strong local restrictions.
Our first construction resolves the large-margin disambiguation question of Alon, Hanneke, Holzman, and Moran~\cite{AHHM22}.
Margin assumptions were a central motivation for their introduction of partial concept classes: leaving labels unspecified near a separating hyperplane allows VC-dimension to be bounded independently of the ambient dimension.
Our result shows that these unspecified labels are essential to this bound: no choice of total extensions can retain a VC-dimension bound independent of the ambient dimension.
Our second construction shows that this remains possible under the strong local requirement that every pure $2\times2$ submatrix be constant.
It also answers further questions about disambiguation under restrictions on online learnability and growth.
We obtain a corresponding obstruction for partial graphs, where symmetry constrains both the construction and its disambiguations.
Our proofs give a different explanation for the increase in complexity.
Earlier constructions force every disambiguation to have many distinct rows, and then deduce large VC-dimension through the Sauer--Shelah lemma~\cite{AHHM22,CHHH23}.
Our approach is structural: Ramsey theory produces structured submatrices inside every disambiguation, within which we explicitly construct large shattered sets.

\smallskip

Partial matrices also arise outside learning theory.
An unspecified entry may represent a deliberate absence of a prescription: for instance, even when all distances are known, we may assign $0$ to pairs of nearby points, $1$ to pairs of distant points, and $\star$ to pairs at intermediate distances.
This interpretation is particularly useful: leaving some entries unspecified can expose structure that is absent from a total representation.

Bourneuf, Charbit, and Thomassé~\cite{BCT25} used this principle of ``proof by ambiguation'': leaving selected pairs of a graph unspecified to obtain a partial graph of bounded VC-dimension, applying tools that exploit this bound, and transferring the resulting conclusions back to the original graph.
Our results concern the reverse operation, and show that the simplicity gained by allowing unspecified entries can be lost under every disambiguation.

More generally, let $p$ be a matrix parameter that does not increase when taking submatrices, and define $p(M)$ for a partial matrix $M$ as the supremum of $p$ over its pure submatrices.
One can ask whether every partial matrix $M$ admits a disambiguation $M'$ with
$p(M')\leq f(p(M))$, for some function $f$.
A simple example shows that the answer is negative for rank over $\mathbb F_2$.

For $n\geq2$, consider the partial matrix $M$ whose rows are indexed by $[n]=\{1,\ldots,n\}$ and whose columns are indexed by pairs $(j,k)$ with $1\leq j<k\leq n$, where
\[
M[i,(j,k)]=
\begin{cases}
0, & i=j,\\
1, & i=k,\\
\star, & \text{otherwise}.
\end{cases}
\]
Two distinct columns are specified on at most one common row, therefore $M$ contains no pure $2\times2$ submatrix and has rank $1$.
In every disambiguation, however, rows $j$ and $k$ differ in column $(j,k)$.
All $n$ rows are therefore distinct, forcing the rank of every disambiguation to be at least $\lceil\log_2 n\rceil$.

This example can also be interpreted in terms of forbidden submatrices.
A standard Ramsey argument shows that a family of binary matrices has bounded rank if and only if, for some $k$, it avoids the three $k \times k$ matrices \[I_k=
\begin{bmatrix}
1 & 0 & \cdots & 0\\
0 & 1 & \ddots & \vdots\\
\vdots & \ddots & \ddots & 0\\
0 & \cdots & 0 & 1
\end{bmatrix},
\qquad
J_k=
\begin{bmatrix}
0 & 1 & \cdots & 1\\
1 & 0 & \ddots & \vdots\\
\vdots & \ddots & \ddots & 1\\
1 & \cdots & 1 & 0
\end{bmatrix},
\qquad 
T_k=
\begin{bmatrix}
1 & 1 & \cdots & 1\\
0 & 1 & \ddots & \vdots\\
\vdots & \ddots & \ddots & 1\\
0 & \cdots & 0 & 1
\end{bmatrix}.\] 
Thus, for sufficiently large $n$, every disambiguation of the preceding matrix contains a prescribed large identity, complemented identity, or triangular submatrix.
Here and throughout, submatrix containment allows independent permutations of rows and columns.

VC-dimension admits a related characterization: a family of binary matrices has bounded VC-dimension if and only if it avoids some fixed binary submatrix.
The disambiguation question can therefore be phrased in terms of forbidden patterns.
Given a binary matrix $F$, is there a binary matrix $F'$ such that every partial matrix containing no pure copy of $F$ admits a disambiguation containing no copy of $F'$?
To express the obstruction to such a statement, we call a family $\mathcal A$ of partial matrices \emph{disambiguation-universal} if, for every finite binary matrix $F$, some $M\in\mathcal A$ has the property that every disambiguation of $M$ contains $F$.

In matrix terms, the construction of Alon, Hanneke, Holzman, and Moran~\cite{AHHM22} gives a disambiguation-universal family $\mathcal G$ avoiding the pure patterns $
\begin{bmatrix}
0&0\\
1&1
\end{bmatrix}$ and $
\begin{bmatrix}
0&1\\
1&0
\end{bmatrix}$.
Their proof draws on a construction of G\"o\"os~\cite{Goos15} producing graphs whose chromatic number is large compared with their biclique partition number.
This separation, arising from the Alon--Saks--Seymour problem, was further improved quantitatively in~\cite{BBGJK21,Pabbaraju26}.

Cheung, Hatami, Hatami, and Hosseini~\cite{CHHH23} established a similar obstruction for online learning: a partial concept class can have Littlestone dimension\footnote{Littlestone dimension is the optimal worst-case mistake bound in deterministic realizable online learning~\cite{AHHM22}.} at most $2$ while every disambiguation has infinite VC-dimension.
Partial classes also retain an important connection between learning models: Fioravanti, Hanneke, Moran, Schefler, and Tsubari~\cite{FHMST24} proved that private learnability implies online learnability in this setting.

The general counterexamples leave open whether such an obstruction can occur for natural geometric classes.
Alon, Hanneke, Holzman, and Moran~\cite{AHHM22} asked whether classes of halfspaces with a fixed positive margin admit disambiguations whose VC-dimension, or Littlestone dimension, is bounded independently of the ambient dimension.
Blondal, Hatami, Hatami, Lalov, and Tretiak~\cite{BHHLT26} answered the Littlestone-dimension question negatively.
The VC-dimension question remained open.
Chornomaz, Moran, and Waknine~\cite{CMW25} related disambiguations of spherical margin classes to simplicial spherical dimension. 
Their correspondence allows the angular parameter to depend on the dimension, whereas our result treats every fixed positive angular parameter.

We consider the following spherical formulation.
For $0<\varepsilon<\pi/2$, let $M_\varepsilon^d$ be the partial matrix indexed by pairs of points of the unit sphere $\mathbb S^d$, where
\[
M_\varepsilon^d[x,y]=
\begin{cases}
0, & d_{\mathbb S}(x,y)\leq\varepsilon,\\
1, & d_{\mathbb S}(x,y)\geq\pi-\varepsilon,\\
\star, & \text{otherwise},
\end{cases}
\]
and $d_{\mathbb S}$ denotes spherical distance.
The classical mistake-bound analysis of the Perceptron algorithm~\cite{Rosenblatt58,Novikoff63} implies that, for each fixed $\varepsilon$, the VC-dimension of $M_\varepsilon^d$ is bounded independently of $d$; see~\cite{AHHM22} for its application to partial concept classes.
Related dimension-independent bounds in other settings appear in~\cite{BCT25,IJN26}.
Moreover, when $\varepsilon<\pi/4$, $M_\varepsilon^d$ contains no pure $2\times2$ submatrix with three equal entries and one of the opposite value.
Indeed, three $0$-entries force the distance corresponding to the fourth entry to be at most $3\varepsilon<\pi-\varepsilon$, by the triangle inequality; the case of three $1$-entries is similar.

\begin{restatable}{theorem}{mainsphere}\label{thm:main-sphere}
    For every $0<\varepsilon<\pi/2$ and $k\in\mathbb N$, every disambiguation of $M_\varepsilon^d$ has VC-dimension at least $k$ for all sufficiently large $d$.
    In particular, the family $\mathcal M_\varepsilon=\{M_\varepsilon^d:d\geq1\}$ is disambiguation-universal.
\end{restatable}

Blondal, Hatami, Hatami, Lalov and Tretiak~\cite{BHHLT26} proved that, for every fixed $\varepsilon$, every disambiguation of $M_\varepsilon^d$ has Littlestone dimension tending to infinity with $d$.
\Cref{thm:main-sphere} strengthens this conclusion to VC-dimension, settling the remaining part of \cite[Open Question 4]{AHHM22}.
We prove \cref{thm:main-sphere} by establishing a corresponding statement on the Hamming cube.

Write $Q_N=\{0,1\}^N$, equipped with Hamming distance $d_H$.
For $0<\varepsilon<1/2$, let $H_\varepsilon^N$ be the \emph{Gap-Hamming partial matrix} whose rows and columns are indexed by $Q_N$, with
\[
H_\varepsilon^N[x,y]=
\begin{cases}
0, & d_H(x,y)\leq\varepsilon N,\\
1, & d_H(x,y)\geq(1-\varepsilon)N,\\
\star, & \text{otherwise}.
\end{cases}
\]

\begin{restatable}[Gap-Hamming Disambiguation]{theorem}{gaphamming}\label{thm:main-cube}
    For every $0<\varepsilon<1/2$ and $k\in\mathbb N$, every disambiguation of $H_\varepsilon^N$ has VC-dimension at least $k$ for all sufficiently large $N$.
\end{restatable}

The standard embedding of the Hamming cube as normalized sign vectors transfers this conclusion to the sphere, yielding \cref{thm:main-sphere}.

The principal Ramsey input is Pálvölgyi's Dense Block theorem~\cite{Palvolgyi26}.
We apply it to obtain a structured restriction on a discrete torus, and show that this structure, together with the gap constraints, forces a large shattered set.
Every even cycle embeds isometrically into a Hamming cube of the same diameter, allowing us to transfer the conclusion from tori to cubes.

Related sign-rank lower bounds for spherical and Gap-Hamming matrices were obtained in~\cite{HHM23,FHV26}.
For VC-dimension, Frick, Hosseini, and Vasileuski~\cite[Conjecture~39]{FHV26} conjecture a lower bound tending to infinity with the distance threshold, uniformly in the ambient dimension.
\Cref{thm:main-cube} establishes the case in which this threshold is a fixed positive proportion of the dimension.

The disambiguation-universal families $\mathcal G$ and $\mathcal M_\varepsilon$ impose different restrictions on their pure $2\times2$ submatrices.
This suggests asking whether a disambiguation-universal family can avoid every nonconstant $2\times2$ pattern.
Such a construction would strengthen the obstruction sought in the forbidden-pattern question of Cheung, Hatami, Hatami, and Hosseini~\cite[Problem~27]{CHHH23}, which singles out the pure pattern
$\begin{bmatrix}1&1\\0&1\end{bmatrix}$.
We obtain this stronger conclusion using the girth Ramsey theorem of Reiher and R\"odl~\cite{RR23}.

\begin{restatable}{theorem}{mainmatrix}\label{thm:matrix-not-disambig}
    For every $k\in\mathbb N$, there exists a finite partial matrix $M$ such that every pure $2\times2$ submatrix of $M$ is constant, while every disambiguation of $M$ has VC-dimension at least $k$.
    In particular, $M$ has VC-dimension at most $1$.
\end{restatable}

The local condition equivalently says that every pure submatrix with at least two rows and two columns is constant.
One cannot additionally forbid a fixed constant matrix while retaining disambiguation-universality.
Indeed, if the forbidden matrix is all $1$, replacing every $\star$ by $0$ preserves its absence; the analogous statement holds for an all-$0$ matrix.

\Cref{thm:main-sphere,thm:main-cube,thm:matrix-not-disambig} also have consequences for online learning and growth.
Alon, Hanneke, Holzman and Moran~\cite[Open Questions 2 and 3]{AHHM22} asked whether polynomial disambiguation growth or finite Littlestone dimension guarantees a disambiguation of finite VC-dimension.
Cheung, Hatami, Hatami and Hosseini~\cite{CHHH23} answered both questions negatively, using a partial concept class of Littlestone dimension at most $2$.
They then asked whether such a disambiguation exists under either of the stronger assumptions of Littlestone dimension $1$ or linear disambiguation growth~\cite[Problems 24 and 25]{CHHH23}.
The following corollary answers both questions negatively.

\begin{corollary}
    There exists a partial concept class of Littlestone dimension $1$ such that every restriction to $n$ domain points admits a disambiguation with at most $n+1$ concepts, whereas every disambiguation of the full class has infinite VC-dimension.
\end{corollary}

Blondal, Hatami, Hatami, Lalov and Tretiak~\cite{BHHLT26} obtained the analogous conclusion for Littlestone dimension using the spherical classes. 

\begin{proof}
    The matrices $M_k$ constructed in \cref{thm:matrix-not-disambig} have Littlestone dimension at most $1$, since all their pure $2 \times 2$ submatrices are constant.
    Taking a countable disjoint union of the matrices $M_k$, with all entries between different blocks left unspecified, yields a partial matrix $M$ of Littlestone dimension at most $1$.
    
    The Standard Optimal Algorithm gives a disambiguation of size at most $\sum_{i=0}^{k}\binom{n}{i}$ for every partial concept class of Littlestone dimension at most $k$ on an $n$-point domain~\cite{AHHM22,CHHH23}.
    Applying this with $k=1$ gives the bound $n+1$.
    
    Every disambiguation of $M$ restricts to a disambiguation of each $M_k$, and hence contains every finite binary matrix.
    Its VC-dimension is therefore infinite.
\end{proof}

The above argument only requires the exclusion of pure $2\times2$ submatrices with three equal entries and one opposite entry. 
Thus the finite examples supplied by \cref{thm:main-sphere,thm:main-cube}, for sufficiently small $\varepsilon$, also yield the corollary.

We finally consider \emph{partial graphs}, in which each unordered pair of vertices is prescribed to be an edge, prescribed to be a non-edge, or left unspecified.
Symmetry imposes additional constraints on the construction, but also restricts the disambiguations it must handle.
We define the VC-dimension of a partial graph using its adjacency matrix with $\star$ on the diagonal.
An induced subgraph is \emph{pure} if all its pairs are specified, and \emph{homogeneous} if it is complete or edgeless.

Without the VC-dimension requirement on the partial graph, the following conclusion follows from a theorem of Ne\v{s}et\v{r}il and R\"odl~\cite{NR81}.

\begin{restatable}{theorem}{maingraph}\label{thm:graph-not-disambig}
    For every $k\in\mathbb N$, there exists a finite partial graph $G$ with VC-dimension at most $1$ such that every pure induced subgraph of $G$ on three vertices is homogeneous, while every disambiguation of $G$ contains every $k$-vertex graph as an induced subgraph.
\end{restatable}

\subsection*{Overview of the proofs}

We first describe a conditional approach to \cref{thm:main-cube}, then explain how a weaker Ramsey principle yields an unconditional proof.
We conclude with the constructions proving \cref{thm:matrix-not-disambig,thm:graph-not-disambig}.

\paragraph{A conditional radial Ramsey approach.}

A disambiguation of $H_\varepsilon^N$ can be viewed as a coloring $c:Q_N\times Q_N\to\{0,1\}$ that assigns color $0$ to every pair at Hamming distance at most $\varepsilon N$, color $1$ to every pair at distance at least $(1-\varepsilon)N$, and is arbitrary on the remaining pairs.
Our goal is to show that every such coloring, viewed as a binary matrix, has large VC-dimension.

For this conditional approach, it suffices to consider symmetric\footnote{A coloring $c : X \times X \to \{0,1\}$ is \emph{symmetric} if $c(x,y)=c(y,x)$ for all $x,y\in X$.} colorings.
Although a disambiguation need not be symmetric, any disambiguation of bounded VC-dimension can be transformed into a symmetric one satisfying the same gap constraints, with VC-dimension bounded by a function of the original one. 

Our starting point is then a natural Ramsey-theoretic strategy: find, inside every symmetric disambiguation, a large copy of a smaller cube with additional structural properties, and use this structure to extract a large shattered set.

Since the Gap-Hamming constraints are metric, the copy should approximately preserve normalized Hamming distances.
Ordinary subcubes are not well suited for this: an $m$-dimensional subcube of $Q_N$ has diameter $m$, and therefore loses the distinction between close and almost-antipodal pairs when $m\ll N$.
Instead, we seek a \emph{near-spanning homothetic copy}, in which distances are scaled uniformly and the diameter is close to that of the ambient cube.
Such a copy can be obtained by replacing each coordinate of $Q_m$ by a block of ambient coordinates, with all blocks of the same size and together covering almost all coordinates of $Q_N$.
The gap constraints then survive inside the copy, with an adjusted parameter.

On such a copy, the desired additional structure is \emph{radiality}: the color of a pair depends only on its Hamming distance.
This would already be enough for our purposes.
Indeed, the gap constraints prescribe one color at small distances and the other at almost-antipodal distances, so the color must change between some intermediate distance levels.
Using such a change of color, one can vary several coordinates independently to construct a large shattered set.

This suggests the following Ramsey-theoretic statement; precise definitions are given in \cref{sec:hdr}.

\begin{restatable}[Near-Spanning Radial Ramsey]{conjecture}{nearspanningradialramsey}\label{conj:radial}
For all $0 < \varepsilon < 1/2$ and $m \in \mathbb{N}$, there exists $N \in \mathbb{N}$ such that every symmetric coloring $c : Q_N \times Q_N \to \{0,1\}$ admits a $c$-radial $\varepsilon$-near-spanning homothetic copy of $Q_m$.
\end{restatable}

The form of \cref{conj:radial} naturally leads to homogeneous dual Ramsey theory.
Indeed, a homothetic copy of $Q_m$ is encoded by a partition of the ambient coordinates into $m$ active blocks of equal size, together with some fixed coordinates. 
A pair of vertices $x,y\in Q_m$ then induces a partition of these blocks into four parts, according to the values taken by $x$ and $y$ on each coordinate.
The sizes of these four parts determine, in particular, the Hamming distance $d_H(x,y)$.
This connects the search for radial copies to Ramsey questions about merging parts of partitions.

The relevant statement is the following conjecture of Kechris, Soki{\'c}, and Todor{\v{c}}evi{\'c}~\cite{KST17}.
Here a $k$-equipartition is a partition into $k$ parts of equal size, and a coarsening is obtained by merging parts.

\begin{restatable}[Homogeneous Dual Ramsey, \cite{KST17}]{conjecture}{homogeneousdualramsey}\label{conj:hdr}
Let $k,m\in\mathbb N$ with $k\mid m$. 
For every $r \in\mathbb N$, there exists an integer $n$, divisible by $m$, such that for every $r$-coloring of the $k$-equipartitions of $[n]$, there exists an $m$-equipartition $\mathcal{P}$ of $[n]$ such that all the $k$-equipartitions coarsening $\mathcal{P}$ have the same color.
\end{restatable}

The correspondence requires some care, since the four classes arising from a pair of cube vertices need not have equal sizes.
In \cref{sec:hdr}, we develop this connection and prove that \cref{conj:hdr} implies \cref{conj:radial}, which in turn implies the Gap-Hamming Disambiguation theorem.

\paragraph{A weaker Ramsey principle.}

For our main theorem, however, we can conclude with a weaker property than the one predicted by \cref{conj:radial}.
The key observation is that shattering is inherently two-sided: the columns to be shattered and the rows witnessing their shattering need not come from the same set of points.
We may therefore restrict the rows and columns along different embeddings $\phi,\psi$, obtaining a coloring of the form $c(\phi(x),\psi(y))$.

We also relax the required dependence on distance.
Rather than asking that the color depend only on the total distance, we allow it to depend on the distances in the individual coordinates.
To make this useful, we work on a discrete torus $(\mathbb Z_R)^N$, a product of cycles of sufficiently large even length $R$.
These distances are naturally captured by the dihedral group: two ordered pairs of vertices of a cycle have the same cyclic distance exactly when a rotation or reflection maps one pair to the other.

The key Ramsey input for our unconditional proof is the recent Dense Block theorem of Pálvölgyi~\cite{Palvolgyi26}.
The two relaxations above allow us to use this theorem in our setting.
We apply it to the dihedral action on ordered pairs and show that its conclusion yields two embeddings $\phi,\psi$ such that $c(\phi(x),\psi(y))$ depends only on the vector of coordinatewise distances between $x$ and $y$.
Moreover, the normalized distance between $\phi(x)$ and $\psi(y)$ differs only slightly from that between $x$ and $y$.
Thus we obtain the desired structured restriction while retaining the gap constraints after a suitable adjustment of the parameter.

The shattering argument is similar in spirit to the radial one: we exploit a change of color forced by the gap constraints to construct a large shattered set.
Here, vectors of coordinatewise distances take the place of total distances, and the extra distance levels on each cycle give us more freedom.

Finally, every even cycle embeds isometrically into a Hamming cube of the same diameter.
We use this embedding to transfer the torus statement back to the cube setting, completing the proof of the Gap-Hamming Disambiguation theorem.

\medskip

\paragraph{The matrix and graph constructions.}

Our constructions for \cref{thm:matrix-not-disambig,thm:graph-not-disambig} use a different Ramsey principle.
We first explain the idea for matrices.

A disambiguation can be viewed as a two-coloring of the unspecified entries.
Given a partial matrix $M$, consider the auxiliary bipartite graph $H$ whose vertices are the rows and columns of $M$, with a row adjacent to a column exactly when the corresponding entry is unspecified.
A disambiguation of $M$ therefore colors the edges of $H$ with $0$ and $1$.
We would like to construct $M$ so that every such coloring contains a monochromatic induced copy of a prescribed bipartite graph $F$, with its two sides corresponding to rows and columns.
Such a copy would identify a submatrix whose unspecified positions are exactly those described by the edges of $F$, and whose unspecified entries all receive the same value.
Inducedness is essential here: it ensures that the nonedges of $F$ between its two sides correspond to specified entries of $M$.

Induced Ramsey theory suggests how to obtain this property.
However, $F$ records only which positions are unspecified, and forgets the values of the specified entries.
We therefore equip $F$ with prescriptions of $0$ or $1$ on the missing pairs between its two sides.
We choose these prescriptions so that the resulting partial pattern only has constant pure $2\times2$ submatrices, but has large VC-dimension whenever its unspecified entries are filled uniformly, whether with $0$ or with $1$.
Placing these prescriptions on the copies of $F$ would then ensure that a monochromatic member yields large VC-dimension.

The difficulty is to make the prescriptions on different copies coexist.
Overlapping copies might prescribe different values to the same entry.
Even if all prescriptions are compatible, four entries contributed by different copies might together form a nonconstant pure $2\times2$ submatrix.
We therefore need a stronger variant of induced Ramsey that controls both how the copies intersect and what small configurations they collectively create.

The girth Ramsey theorem of Reiher and Rödl~\cite{RR23} supplies this additional structure.
It provides a bipartite host and a distinguished family of induced copies of $F$ respecting its two sides, such that every two-coloring of the host edges makes some member of the family monochromatic.
Moreover, every small collection of these copies can be extended to a forest of copies, assembled by taking disjoint unions and gluing along vertices or edges.
In our construction, these intersection properties ensure that no entry receives competing prescriptions.
The local forest structure also allows us to show that every pure $2\times2$ submatrix is contained in a single copy, where it is constant by construction.

For graphs, we apply the same strategy, choosing a partial graph whose two uniform disambiguations each contain every graph of the prescribed order as an induced subgraph.
The local forest structure ensures that every pure triple is homogeneous and that no pair of vertices is shattered.

\subsection*{Preliminaries}

Throughout the paper, for every $n \in \mathbb{N}$, we write $[n]$ for the set $\{1, \ldots, n\}$. 
Given two integers $a, b$, we write $\llbracket a, b\rrbracket$ for the set $\{n \in \mathbb{Z} : a \leq n \leq b\}$.

Let $M$ be a $\{0, 1, \star\}$-matrix.
A set $\mathcal{C}$ of columns of $M$ is \emph{shattered} if for every set $S \subseteq \mathcal{C}$, there exists a row $r_S$ of $M$ such that for every column $c \in \mathcal{C}$, we have $M[r_S, c] = 1$ if $c \in S$, and $M[r_S, c] = 0$ if $c \notin S$.
The \emph{VC-dimension} of $M$ is the supremum of the sizes of finite shattered sets of columns.

\subsection*{AI statement}
The authors suggested Ramsey-theoretic approaches and provided relevant literature.
Building on this input, ChatGPT (OpenAI; GPT-5.6 Sol Pro, GPT-5.6 Sol Ultra, and GPT-6 Astra) was used to discover the proofs of the main results.
It also assisted with further literature searches and the organization and exposition of the manuscript.
The authors take full responsibility for the results, proofs, references, and final text.

\section{Metric obstructions to disambiguation}

We follow the approach described in the proof overview.
We first show that a coloring satisfying a gap property on a suitable discrete torus has large VC-dimension whenever its color depends only on the coordinatewise distances. 
We then use Pálvölgyi’s Dense Block theorem to obtain such a coloring from an arbitrary disambiguation. 
An isometric embedding of even cycles into hypercubes completes the proof of the Gap-Hamming Disambiguation theorem. 
We conclude with the spherical and toroidal consequences.

\subsection{Shattering from distance profiles}

Let $(X, d)$ be a metric space. 
The \emph{diameter} of $(X, d)$ is the supremum of $d(x, y)$ over all $x, y \in X$. We denote it by $\diam_d(X)$, or simply $\diam(X)$ when $d$ is clear from context.
For $0 < \varepsilon < 1/2$, a coloring $c : X \times X \to \{0, 1\}$ is an \emph{$\varepsilon$-gap coloring} if $c(x, y) = 0$ whenever $d(x, y) \leq \varepsilon \cdot \diam(X)$ and $c(x, y) = 1$ whenever $d(x, y) \geq (1-\varepsilon) \cdot \diam(X)$.

Let $Q_n$ denote the $n$-dimensional hypercube $\{0, 1\}^n$, and let $d_{Q_n}$ denote the Hamming distance over $Q_n$.
For $S \subseteq [n]$, we denote by $\mathbbm{1}_S$ the vector $x \in Q_n$ such that $x_i = 1$ if and only if $i \in S$. 
Let $\mathbb{Z}_r = \mathbb{Z}/r\mathbb{Z}$ and for $x, y \in \mathbb{Z}_r$, let $d_{\mathbb{Z}_r}(x, y)$ be the length of the shorter arc between them on the cycle $\mathbb Z_r$. 
Formally, $d_{\mathbb{Z}_r}(x, y) = \min\{|x-y|, r - |x-y|\}$, using representatives in $\llbracket0, r-1\rrbracket$.

Let $(\mathbb Z_r)^n$ denote the $n$-dimensional discrete torus of side length $r$.
For $x, y \in (\mathbb{Z}_r)^n$, the distance between $x$ and $y$ is $d_{(\mathbb{Z}_r)^n}(x, y) = \sum_{i \in [n]}d_{\mathbb{Z}_r}(x_i, y_i)$.
Observe that $\diam((\mathbb{Z}_r)^n) = n \cdot \diam(\mathbb{Z}_r) = n \lfloor r/2 \rfloor$.
For $x, y \in (\mathbb{Z}_r)^n$, we call the vector $\delta(x, y) = (d_{\mathbb{Z}_r}(x_1, y_1), \ldots, d_{\mathbb{Z}_r}(x_n, y_n))$ the \emph{distance profile} of $(x, y)$.

A $2$-coloring $c : (\mathbb{Z}_{r})^n \times (\mathbb{Z}_{r})^n \to \{0, 1\}$ is \emph{distance-profile-canonical} if there exists a function $f : \{0, \ldots, \lfloor r/2 \rfloor\}^n \to \{0, 1\}$ such that for all $x,y \in (\mathbb{Z}_{r})^n$, we have $c(x, y) = f(\delta(x, y))$. In other words, the color of $(x, y)$ only depends on the distance profile of $(x, y)$.

Given a set $V$ and a $2$-coloring $c : V \times V \to \{0, 1\}$, we say that a set $U \subseteq V$ is \emph{shattered} by $c$ if for every subset $W \subseteq U$, there exists $w \in V$ such that for every $u \in U$ we have $c(w, u) = 1$ if and only if $u \in W$.
The \emph{VC-dimension} of $c$ is the supremum of the sizes of its finite shattered sets.
Note that the VC-dimension of $c$ is the VC-dimension of the binary matrix corresponding to $c$.

We first prove that any $\varepsilon$-gap distance-profile-canonical $2$-coloring of a discrete torus of large enough side length and dimension has large VC-dimension.
Intuitively, we choose a $0$-colored profile that is as far as possible from the origin in the torus. Every coordinate that has not yet reached the maximum then gives a direction along which increasing that coordinate changes the color from $0$ to $1$; the gap condition guarantees that there are many such directions. 
The key point is that these many independent color-changing directions can be used to encode arbitrary subsets, and hence to produce a large shattered set.

\begin{lemma}[Shattering from distance profiles]\label{lem:large-vc}
    For all $0 < \varepsilon < 1/2$ and $q \in \mathbb{N}$, there exist an integer $n$ and an even integer $R$ such that every $\varepsilon$-gap distance-profile-canonical coloring $c : (\mathbb{Z}_{R})^n \times (\mathbb{Z}_{R})^n \to \{0, 1\}$ has VC-dimension at least $q$.
\end{lemma}

\begin{proof}
    Let $s$ be large enough so that $\frac{1}{3^s} \leq \varepsilon$.
    Let $r = 3^s$ and $R = 2r$.
    Let $n$ be large enough so that $\varepsilon n \geq q$.
    The factor $3$ comes from a simple distance gadget: $a$ and $3a$ are equally far from $2a$, but have different distances from $0$.
    
    Let $c : (\mathbb{Z}_{R})^n \times (\mathbb{Z}_{R})^n \to \{0, 1\}$ be an $\varepsilon$-gap distance-profile-canonical coloring.
    Since $c$ is distance-profile-canonical, there exists a function $f : \{0, \ldots, r\}^n \to \{0, 1\}$ such that for all $x, y \in (\mathbb{Z}_{R})^n$, we have $c(x, y) = f(\delta(x, y))$.
    
    For every $i \in \{0, \ldots, s\}$, let $d_i = 3^i$; in particular $d_0 = 1$ and $d_s = r$.
    Let $x, y \in (\mathbb{Z}_{R})^n$.
    If $\delta(x, y) = (d_s, \ldots, d_s)$ then $d(x, y) = rn = \diam((\mathbb{Z}_{R})^n)$.
    Since $c$ is an $\varepsilon$-gap coloring, we then have $c(x, y) = 1$, so $f(d_s, \ldots, d_s) = 1$.
    If $\delta(x, y) = (d_0, \ldots, d_0)$ then $d(x, y) = n \leq \varepsilon rn = \varepsilon \cdot  \diam((\mathbb{Z}_{R})^n)$.
    Since $c$ is an $\varepsilon$-gap coloring, we then have $c(x, y) = 0$, so $f(d_0, \ldots, d_0) = 0$.
    
    Consider $a = (a_1, \ldots, a_n) \in \{1, 3, \ldots, r\}^n$ such that $f(a) = 0$ and $\sum_{j = 1}^n a_j$ is maximum with this property.
    Let $J = \{j \in [n] : a_j < r\}$.
    We show that $|J| \geq q$.
    If $x = (0, \ldots, 0)$ then $\delta(x, a) = a$ so \[d(x, a) = \sum_{j = 1}^n a_j \geq \sum_{j \notin J}r + \sum_{j \in J}1 = (n-|J|) \cdot r + |J| = rn - |J| \cdot (r-1).\]
    However, we have $0 = f(a) = f(\delta(x, a)) = c(x, a)$. 
    Since $c$ is an $\varepsilon$-gap coloring, this implies $d(x, a) \leq (1-\varepsilon) \cdot \diam((\mathbb{Z}_{R})^n) = (1-\varepsilon)rn$.
    Thus, $rn - |J| \cdot (r-1) \leq (1-\varepsilon)rn$ so $\varepsilon rn \leq (r-1) |J| \leq r|J|$ so $|J| \geq \varepsilon n \geq q$.
    
    For every $i \in J$, let $x^i \in (\mathbb{Z}_{R})^n$ be the vector such that \[x^i_j = \begin{cases}
      2a_j, &j \in J \setminus \{i\} \\
      0, & \text{ otherwise}.
    \end{cases}\]
    These points are pairwise distinct by definition.
    For every set $S \subseteq J$, let $y^S\in (\mathbb{Z}_{R})^n$ be the vector such that \[y^S_j = \begin{cases}
      3a_j, &j \in S \\
      a_j, &j \notin S.
    \end{cases}\]
    
    Note that for each $j \in J$, we have $a_j \leq d_{s-1} = r/3$ so $0, a_j, 2a_j, 3a_j$ all lie on an arc of length at most $r$ in $\mathbb{Z}_R$. Consequently, $d(0, a_j) = d(a_j, 2a_j) = d(2a_j, 3a_j) = a_j$ and $d(0, 3a_j) = 3a_j$.
    
    We now argue that for every $S \subseteq J$ and every $i \in J$, we have \[\delta(y^S, x^i) = \begin{cases}
        a+2a_ie_i, &i \in S \\
        a, &i \notin S.
    \end{cases}\] 
    By maximality of $a$, this will imply that $c(y^S, x^i) = 1$ if and only if $i \in S$.
    Let $S \subseteq J$ and let $i \in J$.
    Let $j \in [n]$. We consider several cases. \begin{itemize}
        \item If $j \notin J$ then $d(y^S_j, x^i_j) = d(a_j, 0) = d(r, 0) = r = a_j$.
        \item If $j \in J \setminus \{i\}$ then $y^S_j \in \{a_j, 3a_j\}$ and $x^i_j = 2a_j$ so $d(y^S_j, x^i_j) = a_j$.
        \item If $j = i$ and $i \in S$ then $y^S_j = 3a_j$ and $x^i_j = 0$ so $d(y^S_j, x^i_j) = 3a_j$.
        \item If $j=i$ and $i \notin S$ then $y^S_j = a_j$ and $x^i_j = 0$ so $d(y^S_j, x^i_j) = a_j$.
    \end{itemize}
    
    Thus, if $i \in S$, we have $\delta(y^S, x^i) = (a_1, \ldots, a_{i-1}, 3a_i, a_{i+1}, \ldots, a_n)$, so by maximality of $(a_1, \ldots, a_n)$, we have $c(y^S, x^i) = f(a_1, \ldots, a_{i-1}, 3a_i, a_{i+1}, \ldots, a_n) = 1$.
    However, if $i \notin S$, we have $\delta(y^S, x^i) = (a_1, \ldots, a_n) = a$ so $c(y^S, x^i) = f(a) = 0$.
    
    Therefore, the set $\{x^i : i \in J\}$ is shattered by $c$, so $c$ has VC-dimension at least $q$.
\end{proof}

The following two lemmas describe how gap constraints and VC-dimension behave when we pull back the two arguments of a coloring along possibly different maps.

\begin{lemma}[Gap preservation under pullback]\label{lem:transfer-gap-coloring}
    Let $0 < \varepsilon < 1/2$ and let $(X, d_X), (Y, d_Y)$ be two finite metric spaces such that there exist two maps $\iota_1, \iota_2 : Y \to X$ that satisfy $d_X(\iota_1(y), \iota_2(y')) = \lambda \cdot d_Y(y, y') + \Delta$ for all $y, y' \in Y$, for some $\lambda, \Delta \geq 0$.
    Let $c : X \times X \to \{0, 1\}$ be an $\varepsilon$-gap coloring and define $\widehat{c} : Y \times Y \to \{0, 1\}, (y, y') \mapsto c(\iota_1(y), \iota_2(y'))$.
    Suppose that $\lambda \cdot \diam(Y) \geq (1-\varepsilon/2) \cdot \diam(X)$.
    Then, $\widehat{c}$ is an $\varepsilon/2$-gap coloring.
\end{lemma}

\begin{proof}
    Let $y, y' \in Y$ such that $d_Y(y, y') = \diam(Y)$.
    Then, \[d_X(\iota_1(y), \iota_2(y')) = \lambda \cdot d_Y(y, y') + \Delta \geq \lambda \cdot \diam(Y),\] so $\diam(X) \geq \lambda \cdot \diam(Y)$.
    Furthermore, we have \[\diam(X) \geq d_X(\iota_1(y), \iota_2(y')) = \lambda \cdot d_Y(y, y') + \Delta = \lambda \cdot \diam(Y) + \Delta \geq (1-\varepsilon/2) \cdot \diam(X) + \Delta,\] so $\Delta \leq \varepsilon/2 \cdot \diam(X)$.

    Let $y, y' \in Y$ such that $d_Y(y, y') \leq \varepsilon/2 \cdot \diam(Y)$.
    Then, \begin{align*}
        d_X(\iota_1(y), \iota_2(y')) &= \lambda \cdot d_Y(y, y') + \Delta \leq \lambda \cdot \varepsilon/2 \cdot \diam(Y) + \varepsilon/2 \cdot \diam(X) \\
        &\leq \varepsilon/2 \cdot \diam(X) + \varepsilon/2 \cdot \diam(X) = \varepsilon \cdot \diam(X).
    \end{align*} Thus, since $c$ is an $\varepsilon$-gap coloring, we have $\widehat{c}(y, y') = c(\iota_1(y), \iota_2(y')) = 0$.
    \smallskip
    
    Let $y, y' \in Y$ such that $d_Y(y, y') \geq (1-\varepsilon/2) \cdot \diam(Y)$.
    Then, \begin{align*}
        d_X(\iota_1(y), \iota_2(y')) &= \lambda \cdot d_Y(y, y') + \Delta \geq \lambda \cdot (1-\varepsilon/2) \cdot \diam(Y) \\
        &\geq (1-\varepsilon/2) \cdot (1-\varepsilon/2) \cdot \diam(X) \geq (1 - \varepsilon) \cdot \diam(X).
    \end{align*} 
    Thus, since $c$ is an $\varepsilon$-gap coloring, we have $\widehat{c}(y, y') = c(\iota_1(y), \iota_2(y')) = 1$.
    This concludes the proof that $\widehat{c}$ is an $\varepsilon/2$-gap coloring.
\end{proof}

\begin{lemma}[VC-dimension preservation under pullback]\label{lem:transfer-VC-dim}
    Let $X, Y$ be two sets and let $\iota_1, \iota_2 : Y \to X$ be two maps.
    Let $c : X \times X \to \{0, 1\}$ be a coloring and define $\widehat{c} : Y \times Y \to \{0, 1\}, (y, y') \mapsto c(\iota_1(y), \iota_2(y'))$.
    If $\widehat{c}$ has VC-dimension at least $q \in \mathbb{N}$ then so does $c$.
\end{lemma}

\begin{proof}
    Suppose that $\widehat{c}$ has VC-dimension at least $q$.
    Let $\{y^i : i \in [q]\} \subseteq Y$ be shattered by $\widehat{c}$: for every $S \subseteq [q]$, there exists $y^S \in Y$ such that for every $i \in [q]$, we have $\widehat{c}(y^S, y^i) = 1$ if and only if $i \in S$.
    By definition of $\widehat{c}$, this means that for every $S \subseteq [q]$ and every $i \in [q]$, we have $c(\iota_1(y^S), \iota_2(y^i)) = 1$ if and only if $i \in S$.
    Thus, the set $\{\iota_2(y^i) : i \in [q]\}$ is shattered by $c$.
    In particular, this implies that its elements are pairwise distinct, so $c$ has VC-dimension at least $q$.
\end{proof}

\subsection{Canonicalization on a torus}

In view of \cref{lem:large-vc}, we would like to find a restriction on which the color depends only on the orbit of each coordinate pair. 
We also need the variable coordinates to occupy almost the entire ambient space, so that the gap condition survives the restriction.
Our key tool to find this is a result of Pálvölgyi \cite{Palvolgyi26}. First, we need to introduce some notation.

Let $G$ be a finite group acting on a finite set $X$.
A \emph{$G$-block map} of \emph{dimension} $m$, \emph{length} $n$, and \emph{block size} $k$ is a map $\phi : X^m \to X^n$ for which there are pairwise disjoint sets $I_1, \ldots, I_m \subseteq [n]$, each of cardinality $k$, labels $\gamma_{\ell} \in G$ for $\ell \in I_1 \cup \ldots \cup I_m$, and fixed letters $z_\ell \in X$ outside this union, such that for every $x \in X^m$, we have \[
  \phi(x)_\ell=
  \begin{cases}
    \gamma_\ell x_j,&\ell\in I_j,\\
    z_\ell,&\ell\notin I_1\cup\cdots\cup I_m.
  \end{cases}
\]
The sets $I_1,\ldots,I_m$ are the \emph{blocks} of the map, and $mk/n$ is its \emph{active proportion}.

The action of $G$ on $X$ has the \emph{dense block property} if for every $r, m \in \mathbb{N}$ and every $0 < \eta < 1$, there are $n, k \in \mathbb{N}$ such that $mk \geq (1 - \eta)n$, and every coloring $c : X^n \to [r]$ admits a $G$-block map $\phi : X^m \to X^n$ of block size $k$ satisfying \[c(\phi(x)) = c(\phi(y)) \text{ for all $x, y \in X^m$ such that $y_j \in G \cdot x_j$ for every $j \in [m]$.}\]

Said differently, the induced coloring $c\circ\phi$ depends only on the $G$-orbit of each input coordinate, while the active blocks cover at least a $1-\eta$ proportion of the ambient coordinates.
For instance, observe that if the action of $G$ on $X$ is transitive, then $c(\phi(x))$ is constant over all $x \in X^m$.

The following theorem of Pálvölgyi shows that solvability of the acting group suffices to guarantee the dense block property.

\begin{theorem}[Pálvölgyi’s Dense Block theorem {\cite[Theorem 11]{Palvolgyi26}}]\label{lem:palvolgyi}
    Every action of a finite solvable group on a finite set has the dense block property.
\end{theorem}

We apply this theorem to the diagonal action of the dihedral group on pairs of cycle vertices to obtain a distance-profile-canonical restriction while preserving the gap constraints.

\begin{lemma}[Canonicalization on a torus]\label{lem:make-nice}
    For all $0 < \varepsilon < 1/2$ and $R, n \in \mathbb{N}$, there exists $N \in \mathbb{N}$ with the following property.
    For every $\varepsilon$-gap coloring $c : (\mathbb{Z}_R)^N \times (\mathbb{Z}_R)^N \to \{0, 1\}$, there exist two maps $\iota_1, \iota_2 : (\mathbb{Z}_R)^n \to (\mathbb{Z}_R)^N$ such that the coloring $\widehat{c} : (\mathbb{Z}_{R})^n \times (\mathbb{Z}_{R})^n \to \{0, 1\}$ defined by $\widehat{c}(x, y) = c(\iota_1(x), \iota_2(y))$ is an $\varepsilon/2$-gap distance-profile-canonical coloring.  
\end{lemma}

Note that allowing different row and column embeddings lets us obtain a distance-profile-canonical restriction without first symmetrizing the original coloring.

\begin{proof}
    Consider the dihedral group $D_R$ (the group of symmetries of the $R$-gon). 
    Note that $D_R \cong C_R \rtimes C_2$ is a finite solvable group.
    By identifying $\mathbb{Z}_R$ with the $R$-gon, $D_R$ naturally acts on $\mathbb{Z}_R$.
    Consider the diagonal action of $D_R$ on $\mathbb{Z}_R \times \mathbb{Z}_R$, defined by $g \cdot (x, y) = (g \cdot x, g \cdot y)$.
    Let $N, k$ be given by \cref{lem:palvolgyi} for this action of $D_R$ on $\mathbb{Z}_R \times \mathbb{Z}_R$, for $r=2$, $n$ and $\eta = \varepsilon/2$.
    In particular, we have $nk \geq (1-\varepsilon/2)N$.
    
    For every $m \geq 1$, we identify the sets $(\mathbb{Z}_R)^m \times (\mathbb{Z}_R)^m$ and $(\mathbb{Z}_R \times \mathbb{Z}_R)^m$ via the bijection \[((x_1, \ldots, x_m), (y_1, \ldots, y_m)) \mapsto ((x_1, y_1), \ldots, (x_m, y_m)).\]
    Let $c : (\mathbb{Z}_R)^N \times (\mathbb{Z}_R)^N \to \{0, 1\}$ be an $\varepsilon$-gap coloring.
    We view $c$ as a $2$-coloring of $(\mathbb{Z}_R \times \mathbb{Z}_R)^N$.
    Applying \cref{lem:palvolgyi} to the diagonal action of $D_R$ on $\mathbb{Z}_R \times \mathbb{Z}_R$ for $r=2$, $n$ and $\eta = \varepsilon/2$, we obtain a $D_R$-block map $\phi : (\mathbb{Z}_R \times \mathbb{Z}_R)^n \to (\mathbb{Z}_R \times \mathbb{Z}_R)^N$ of block size $k$ such that for all $(x, y), (x', y') \in (\mathbb{Z}_R \times \mathbb{Z}_R)^n$, we have $c(\phi(x, y)) = c(\phi(x', y'))$ whenever $(x'_i, y'_i) \in D_{R} \cdot (x_i, y_i)$ for every $i \in [n]$.
    
    Let $I_1, \ldots, I_n \subseteq [N]$ be the corresponding blocks, recall that each has cardinality $k$, let $g_{t} \in D_R$ for $t \in I_1 \cup \ldots \cup I_n$ be the corresponding labels, and $z_{t} = (a_t, b_t) \in \mathbb{Z}_R \times \mathbb{Z}_R$ for $t$ outside of this union be the corresponding fixed letters.
    
    Observe that for $(a, b), (a', b') \in \mathbb{Z}_R \times \mathbb{Z}_R$, we have $(a', b') \in D_R \cdot (a, b)$ if and only if $d_{\mathbb{Z}_R}(a, b) = d_{\mathbb{Z}_R}(a', b')$.
    Indeed, after rotating $a$ to $a'$, a reflection, if necessary, matches the orientation of the displacement from $a$ to $b$ with that from $a'$ to $b'$.
    Thus, for all $x, y, x', y' \in (\mathbb{Z}_R)^n$, we have $c(\phi(x, y)) = c(\phi(x', y'))$ whenever $\delta(x, y) = \delta(x', y')$.
    
    Since the action of $D_R$ on $\mathbb{Z}_R \times \mathbb{Z}_R$ is diagonal, the map $\phi : (\mathbb{Z}_R \times \mathbb{Z}_R)^n \to (\mathbb{Z}_R \times \mathbb{Z}_R)^N$ factors into two maps $\iota_1, \iota_2 : (\mathbb{Z}_R)^n \to (\mathbb{Z}_R)^N$ so that $\phi(x, y) = (\iota_1(x), \iota_2(y))$ for all $x, y \in (\mathbb{Z}_R)^n$.
    
    Thus, for all $x, y, x', y' \in (\mathbb{Z}_R)^n$ such that $\delta(x, y) = \delta(x', y')$, we have \[\widehat{c}(x, y) = c(\iota_1(x), \iota_2(y)) = c(\phi(x, y)) = c(\phi(x', y')) = c(\iota_1(x'), \iota_2(y')) = \widehat{c}(x', y').\]
    This proves that $\widehat{c}$ is distance-profile-canonical.
    
    Finally, let $x, y \in (\mathbb{Z}_R)^n$.
    Set $\Delta = \sum_{t \notin \bigcup_{j \in [n]}I_j}d_{\mathbb{Z}_R}(a_t, b_t)$.
    Using that $D_R$ acts on $\mathbb{Z}_R$ by isometries, we have \begin{align*}
        d(\iota_1(x), \iota_2(y)) &= \sum_{t \in [N]}d_{\mathbb{Z}_R}(\iota_1(x)_t, \iota_2(y)_t) \\
            &= \sum_{j \in [n]}\sum_{t \in I_j} d_{\mathbb{Z}_R}(g_t \cdot x_j, g_t \cdot y_j) + \sum_{t \notin \bigcup_{j \in [n]}I_j}d_{\mathbb{Z}_R}(a_t, b_t) \\
            &= \sum_{j \in [n]}\sum_{t \in I_j} d_{\mathbb{Z}_R}(x_j, y_j) + \Delta \\
            &= k \sum_{j \in [n]}d_{\mathbb{Z}_R}(x_j, y_j) + \Delta \\
            &= k \cdot d(x, y) + \Delta.
    \end{align*} 
    We want to apply \cref{lem:transfer-gap-coloring} to $((\mathbb{Z}_R)^N, d_{(\mathbb{Z}_R)^N})$ and $((\mathbb{Z}_R)^n, d_{(\mathbb{Z}_R)^n})$ with the maps $\iota_1, \iota_2 : (\mathbb{Z}_R)^n \to (\mathbb{Z}_R)^N$.
    For this, we need to verify that $k \cdot \diam((\mathbb{Z}_R)^n) \geq (1 - \varepsilon/2) \cdot \diam((\mathbb{Z}_R)^N)$.
    Indeed, we have \[k \cdot \diam((\mathbb{Z}_R)^n) = kn \cdot \diam(\mathbb{Z}_R) \geq (1 - \varepsilon/2) \cdot N \cdot \diam(\mathbb{Z}_R) = (1 - \varepsilon/2) \cdot \diam((\mathbb{Z}_R)^N).\]
    Thus, $\widehat{c}$ is an $\varepsilon/2$-gap coloring by \cref{lem:transfer-gap-coloring}.
\end{proof}

\subsection{From tori to the hypercube}

Combining the canonicalization in \cref{lem:make-nice} with the shattering argument of \cref{lem:large-vc}, we obtain large VC-dimension for every gap coloring of a suitable discrete torus.

\begin{theorem}\label{thm:exists-torus}
    For all $0 < \varepsilon < 1/2$ and $q \in \mathbb{N}$, there exist an integer $N$ and an even integer $R$ such that every $\varepsilon$-gap coloring $c : (\mathbb{Z}_R)^N \times (\mathbb{Z}_R)^N \to \{0, 1\}$ has VC-dimension at least $q$.
\end{theorem}

\begin{proof}
    Given $0 < \varepsilon < 1/2$ and $q \in \mathbb{N}$, let $n$ be the integer and $R$ be the even integer given by \cref{lem:large-vc} for $\varepsilon/2$ and $q$, and let $N$ be given by \cref{lem:make-nice} for $\varepsilon, R$ and $n$.
    
    Let $c : (\mathbb{Z}_R)^N \times (\mathbb{Z}_R)^N \to \{0, 1\}$ be an $\varepsilon$-gap coloring.
    By \cref{lem:make-nice}, there exist two maps $\iota_1, \iota_2 : (\mathbb{Z}_R)^n \to (\mathbb{Z}_R)^N$ such that the coloring $\widehat{c} : (\mathbb{Z}_{R})^n \times (\mathbb{Z}_{R})^n \to \{0, 1\}$ defined by $\widehat{c}(x, y) = c(\iota_1(x), \iota_2(y))$ is an $\varepsilon/2$-gap distance-profile-canonical coloring.
    By \cref{lem:large-vc}, the coloring $\widehat{c}$ has VC-dimension at least $q$.
    By \cref{lem:transfer-VC-dim}, $c$ also has VC-dimension at least $q$.
\end{proof}

Using the isometric embedding of $\mathbb Z_{2r}$ into $Q_r$, we transfer this result to the hypercube and obtain \cref{thm:main-cube}, which we restate for convenience.

\gaphamming*

\begin{proof}
    Given $0 < \varepsilon < 1/2$ and $k \in \mathbb{N}$, let $n$ be the integer and $R$ be the even integer given by \cref{thm:exists-torus} for $\varepsilon/2$ and $k$.
    Let $r = R/2$ and set $N_0=\lceil 2rn/\varepsilon\rceil$.
    Let $N \geq N_0$ and consider an arbitrary disambiguation of $H_{\varepsilon}^N$.
    This disambiguation naturally corresponds to an $\varepsilon$-gap coloring $c : Q_N \times Q_N \to \{0, 1\}$.
    Write $N=arn+b$, where $0\leq b<rn$.
    Since $N\geq N_0$, we have 
    \[b<rn\leq\frac{\varepsilon N}{2} \qquad \text{ and } \qquad arn=N-b\geq(1-\varepsilon/2)N.\]
    
    The key observation is that the vertices $0^r, 10^{r-1}, 110^{r-2}, \ldots, 1^{r-1}0, 1^r, 01^{r-1}, 001^{r-2}, \ldots, 0^{r-1}1$ form an isometric cycle of length $2r$ in $Q_r$.
    Indeed, if $\varphi(t)$ denotes the $t$-th listed vertex, then $d_{Q_r}(\varphi(s), \varphi(t)) = \min\{|s-t|, 2r-|s-t|\} = d_{\mathbb{Z}_{2r}}(s, t)$.
    This gives an isometry $\varphi : \mathbb{Z}_R \to Q_r$, which naturally extends to an isometry $\phi : (\mathbb{Z}_R)^{n} \to (Q_r)^{n}, x \mapsto (\varphi(x_1), \ldots, \varphi(x_{n}))$.
    Identifying $((Q_r)^{n})^a \times Q_b$ with $Q_{arn+b} = Q_N$, we obtain a map \[\psi : (\mathbb{Z}_R)^n \to Q_N, x \mapsto (\underbrace{\phi(x),\ldots,\phi(x)}_{a\text{ times}},\underbrace{0,\ldots,0}_{b\text{ times}}).
    \]
    Then, for all $x, y \in (\mathbb{Z}_R)^n$, we have $d(\psi(x), \psi(y)) = a \cdot d(x, y)$.
    
    We want to apply \cref{lem:transfer-gap-coloring} to $(Q_N, d_{Q_N})$ and $((\mathbb{Z}_R)^n, d_{(\mathbb{Z}_R)^n})$ with the maps $\psi, \psi : (\mathbb{Z}_R)^n \to Q_N$.
    For this, we need to verify that $a \cdot \diam((\mathbb{Z}_R)^n) \geq (1 - \varepsilon/2) \cdot \diam(Q_N)$.
    Indeed, we have \[a \cdot \diam((\mathbb{Z}_R)^n) = arn \geq (1-\varepsilon/2) \cdot N = (1 - \varepsilon/2) \cdot \diam(Q_N).\]
    Thus, the coloring $\widehat{c} : (\mathbb{Z}_R)^n \times (\mathbb{Z}_R)^n \to \{0, 1\}$ defined by $\widehat{c}(x, y) = c(\psi(x), \psi(y))$ is an $\varepsilon/2$-gap coloring of $(\mathbb{Z}_R)^n$.
    By our choice of $n$ and by \cref{thm:exists-torus}, we obtain that $\widehat{c}$ has VC-dimension at least $k$.
    Then, $c$ has VC-dimension at least $k$ by \cref{lem:transfer-VC-dim}.
    This proves that the corresponding disambiguation of $H_{\varepsilon}^N$ has VC-dimension at least $k$.
\end{proof}

\subsection{Spherical and toroidal consequences}

Recall that, for $0<\varepsilon<\pi/2$ and $d \in \mathbb{N}$, $M_\varepsilon^d$ is the partial matrix indexed by pairs of points of the unit sphere $\mathbb S^d$, where
\[
M_\varepsilon^d[x,y]=
\begin{cases}
0, & d_{\mathbb S}(x,y)\leq\varepsilon,\\
1, & d_{\mathbb S}(x,y)\geq\pi-\varepsilon,\\
\star, & \text{otherwise},
\end{cases}
\]
and $d_{\mathbb S}$ denotes spherical distance.

As explained in the introduction, \cref{thm:main-sphere} follows quickly from \cref{thm:main-cube}. We restate it for convenience.

\mainsphere*

\begin{proof}
    Let $0 < \varepsilon < \pi/2$ and $k \geq 1$ be given.
    Set $\eta = \frac{1-\cos\varepsilon}{2}$.
    Let $d$ be large enough so that every disambiguation of $H_{\eta}^{d+1}$ has VC-dimension at least $k$; this is possible by \cref{thm:main-cube}.
    Let \[\phi : Q_{d+1} \to \mathbb{S}^{d}, x \mapsto \frac{1}{\sqrt{d+1}}\left((-1)^{x_1}, \ldots, (-1)^{x_{d+1}}\right).\]
    Then, for all $x, y \in Q_{d+1}$, we have $\langle \phi(x), \phi(y)\rangle = 1-\frac{2d(x, y)}{d+1}$ so $d_{\mathbb{S}^{d}}(\phi(x), \phi(y)) = \arccos\left(1-\frac{2d(x, y)}{d+1}\right)$.
    Thus, if $d(x, y) \leq \eta (d+1)$ then $d_{\mathbb{S}^{d}}(\phi(x), \phi(y)) \leq \varepsilon$ and if $d(x, y) \geq (1-\eta) (d+1)$ then $d_{\mathbb{S}^{d}}(\phi(x), \phi(y)) \geq \pi - \varepsilon$.
    
    Any disambiguation of $M_{\varepsilon}^d$ gives a coloring $c : \mathbb{S}^{d} \times \mathbb{S}^{d} \to \{0, 1\}$.
    By the above, the coloring $\widehat{c} : Q_{d+1} \times Q_{d+1} \to \{0, 1\}, (x, y) \mapsto c(\phi(x), \phi(y))$ is an $\eta$-gap coloring.
    By \cref{thm:main-cube}, $\widehat{c}$ has VC-dimension at least $k$.
    By \cref{lem:transfer-VC-dim}, $c$ also has VC-dimension at least $k$.
    In other words, any disambiguation of $M_{\varepsilon}^d$ has VC-dimension at least $k$.
\end{proof}
 
The intermediate torus theorem (\cref{thm:exists-torus}) allowed us to choose both the side length and the dimension.
Having established the Gap-Hamming Disambiguation theorem, we can now bootstrap this result to obtain a statement uniform over all side lengths. 
The reason is that every discrete torus contains a scaled isometric copy of the hypercube with the same diameter.

\begin{corollary}\label{cor:main-torus}
    For all $0 < \varepsilon < 1/2$ and $q \in \mathbb{N}$, there exists an integer $N_0$ such that for every $N \geq N_0$ and every integer $R \geq 2$, every $\varepsilon$-gap coloring $c : (\mathbb{Z}_R)^N \times (\mathbb{Z}_R)^N \to \{0, 1\}$ has VC-dimension at least $q$.
\end{corollary}

\begin{proof}
    Given $0 < \varepsilon < 1/2$ and $q \in \mathbb{N}$, let $N_0 \in \mathbb{N}$ be given by \cref{thm:main-cube} for $\varepsilon$ and $q$.
    Let $N \geq N_0$, let $R \geq 2$ and let $c : (\mathbb{Z}_R)^N \times (\mathbb{Z}_R)^N \to \{0, 1\}$ be an $\varepsilon$-gap coloring.
    Set $r = \lfloor R/2\rfloor$ and let \[\iota : Q_N \to (\mathbb{Z}_R)^N, x \mapsto (rx_1, \ldots, rx_N).\]
    Then, for all $x, y \in Q_N$, we have $d(\iota(x), \iota(y)) = rd(x, y)$.
    Furthermore, $r \cdot \diam(Q_N) = rN = \diam((\mathbb{Z}_R)^N)$ so the coloring \[\widehat{c} : Q_N \times Q_N \to \{0, 1\}, (x, y) \mapsto c(\iota(x), \iota(y))\] is an $\varepsilon$-gap coloring.
    By \cref{thm:main-cube}, $\widehat{c}$ has VC-dimension at least $q$.
    By \cref{lem:transfer-VC-dim}, $c$ also has VC-dimension at least $q$.
\end{proof}

\section{Girth Ramsey and Disambiguations}

In this section, we prove
\cref{thm:matrix-not-disambig,thm:graph-not-disambig} following the approach outlined in the proof overview.

We first present the girth Ramsey theorem and establish the localization properties needed for both constructions.
We then treat matrices and graphs separately, choosing the appropriate partial pattern within each construction.
For matrices, both uniform fillings of the pattern have large VC-dimension.
For graphs, both uniform fillings contain every graph of the prescribed order as an induced subgraph.
In each case, we apply the girth Ramsey theorem to arrange copies of this pattern so that every disambiguation completes one copy uniformly.
The localization properties ensure that the prescriptions are compatible and preserve the required local restrictions.

\subsection{The girth-Ramsey theorem}

The induced Ramsey theorem states that for every graph $F$, there exists a graph $H$ such that every $2$-coloring of the edges of $H$ contains a monochromatic induced copy of $F$.
It was proved independently by Deuber~\cite{Deuber75}, Erd\H{o}s, Hajnal, and P\'osa~\cite{EHP75}, and R\"odl~\cite{Rodl73}.

As discussed in the proof overview, our constructions require additional control over how these copies intersect and what small configurations they collectively form.
This control is needed both to avoid conflicting prescriptions and to preserve our local restrictions.
We therefore use the girth Ramsey theorem of Reiher and R\"odl~\cite{RR23}, which strengthens the induced Ramsey theorem by providing additional control over the local structure of a family of copies.
To state the form we need, we first introduce the notion of a forest of copies.

A family $\mathcal{F}$ of copies of a graph $F$ in a graph $H$ is a \emph{forest of copies of $F$} if its members can be ordered as $F_1, \ldots, F_m$ so that for every $j \geq 1$, the set $V(F_j) \cap \left(\bigcup_{i < j}V(F_i)\right)$ is either empty, a single vertex, or the two ends of an edge of $H$ belonging both to $F_j$ and some earlier $F_i$.

\begin{theorem}[Girth Ramsey {\cite{RR23}}]\label{thm:girth-ramsey}
    Let $F$ be a graph and let $r, \ell \in \mathbb{N}$.
    There exists a graph $H$ and a set $\mathcal{F}$ of induced copies of $F$ in $H$ with the following properties: \begin{itemize}
        \item Every family $\mathcal{S} \subseteq \mathcal{F}$ of size at most $\ell$ is contained in a forest $\mathcal{S'} \subseteq \mathcal{F}$ of copies of $F$.
        \item For every coloring $\chi : E(H) \to [r]$, some member of $\mathcal{F}$ is monochromatic.
    \end{itemize}
    Moreover, if $F$ is equipped with a bipartition then $H$ may be chosen bipartite and the members of $\mathcal{F}$ may be required to respect the two sides.
\end{theorem}

The first assertion is the graph case of the girth-Ramsey theorem \cite[Theorem 1.7]{RR23}. 
For the bipartite refinement, regard the prescribed sides as an $f$-partition with $f=(1,1)$.
In the proof of \cite[Theorem 13.12]{RR23}, choose the initial Ramsey construction in the partite form supplied by \cite[Proposition 3.8]{RR23}. 
\cite[Section 3.6]{RR23} shows that the subsequent construction preserves the prescribed classes and the side-respecting copies.

The formulation “is contained in a forest of copies”, rather than “is itself a forest of copies”, is necessary because the property of being a forest of copies is not inherited by subfamilies, see \cite{RR23}.

\subsection{Localization in forests of copies}

The localization arguments below use the fact that a forest of copies is assembled along sets of at most two vertices.
We make this structure explicit by replacing each copy by a clique and describing the resulting graph by a tree-decomposition whose bags are precisely the vertex sets of the copies.
For this, we first need some definitions.

A \emph{tree-decomposition} of a graph $G$ is a pair $(T, \mathcal{V})$ where $T$ is a tree and $\mathcal{V} = (V_t)_{t \in V(T)}$ is a family of subsets of $V(G)$, called \emph{bags}, with the following properties: \begin{itemize}
    \item For every $v \in V(G)$, the set $T_v = \{t \in V(T) : v \in V_t\}$ induces a nonempty subtree of $T$, and
    \item For every edge $uv \in E(G)$, there exists a node $t \in V(T)$ such that $\{u, v\} \subseteq V_t$.
\end{itemize}

Let $(T, \mathcal{V})$ be a tree-decomposition of a graph $G$.
Given an edge $e = tt'$ of $T$, the set $V_t \cap V_{t'}$ is the \emph{adhesion set} at $e$.
The size of a largest adhesion set is the \emph{adhesion} of $(T, \mathcal{V})$.
We will later use the following simple property of tree-decompositions: Let $e$ be an edge of $T$ and let $T_1, T_2$ be the two connected components of $T-e$. Let $V_1 = \bigcup_{t \in V(T_1)}V_t$ and $V_2 = \bigcup_{t \in V(T_2)}V_t$. If $w \in V(G)$ has a neighbor in $V_1 \setminus V_2$ and a neighbor in $V_2 \setminus V_1$ then $w$ belongs to the adhesion set at $e$. 

\medskip

If $\mathcal{F}$ is a nonempty forest of copies of a graph $F$ in a graph $H$, we define the graph $\Gamma(\mathcal{F})$ on vertex set $\bigcup_{Q \in \mathcal{F}}V(Q)$, where two distinct vertices $u, v \in V(\Gamma(\mathcal{F}))$ are adjacent if and only if they belong to a common member of $\mathcal{F}$.

\begin{observation}
    For every nonempty forest of copies of a graph $F$ in a graph $H$, the graph $\Gamma(\mathcal{F})$ has a tree-decomposition of adhesion at most $2$, where every adhesion set of size $2$ induces an edge of $H$, and whose bags are precisely the sets $V(Q)$ for $Q\in\mathcal F$. \\
    We call it the \emph{standard tree-decomposition} of $\Gamma(\mathcal{F})$.
\end{observation}

\begin{proof}
    Fix a forest ordering $F_1, \ldots, F_m$ of $\mathcal{F}$.
    We build a tree $T$ whose vertices are the members of $\mathcal{F}$, as follows.
    Start from the tree with the unique vertex $F_1$.
    Iteratively, for every $2 \leq j \leq m$ add the vertex $F_j$ to $T$ and attach it to an arbitrary vertex $F_k$ with $k < j$ such that $V(F_j) \cap (\bigcup_{i < j}V(F_i)) \subseteq V(F_k)$; such a vertex exists since $\mathcal{F}$ is a forest of copies of $F$.
    Let $\mathcal{V} = (V(F_j))_{F_j \in \mathcal{F}}$.
    Since $\mathcal{F}$ is a forest of copies of $F$ in $H$, it follows that $(T, \mathcal{V})$ is a tree-decomposition of $\Gamma(\mathcal{F})$ of adhesion at most $2$, and every adhesion set of size $2$ induces an edge of $H$.
    Indeed, when $F_j$ is attached to $F_k$, every vertex of $F_j$ occurring in an earlier copy also belongs to $F_k$. 
    Thus the bags containing any fixed vertex remain connected. 
    Every edge of $\Gamma(\mathcal F)$ lies in a bag by definition.
\end{proof}

The intersection property below will ensure that prescriptions on different copies cannot conflict.
The remaining statements will allow us to verify our local restrictions by reducing potential violations to configurations within individual copies.

\begin{lemma}[Localization in forests of copies]\label{lem:properties-forest}
    Let $\mathcal{F}$ be a forest of copies of a graph $F$ in a graph $H$.
    Then the following statements hold. \begin{enumerate}[label=\alph*)]
        \item\label{itm:forest-no-nonedge} Any two distinct members $Q,Q'$ intersect in at most one vertex or in the two ends of an edge of $H$.
        \item\label{itm:clique-common-member} Every clique in $\Gamma(\mathcal{F})$ is contained in a common member of $\mathcal{F}$.
        \item\label{itm:3-nbrs-common-member} If two vertices $u$ and $v$ have at least three common neighbors in $\Gamma(\mathcal{F})$ then they occur in a common member of $\mathcal{F}$. Moreover, every common neighbor $w$ of $u$ and $v$ occurs together with both $u$ and $v$ in a member of $\mathcal{F}$.
        \item\label{itm:C4-common-member} If $u_1u_2u_3u_4$ is a $4$-cycle in $\Gamma(\mathcal{F})$ and $u_1u_3$ and $u_2u_4$ are non-edges of $H$ then some member of $\mathcal{F}$ contains $u_1, u_2, u_3$ and $u_4$.
    \end{enumerate}
\end{lemma}

\begin{proof}
    \ref{itm:forest-no-nonedge} Let $F_i,F_j$ be distinct members of $\mathcal{F}$ with $i<j$ in a forest ordering. Then
        \[V(F_j)\cap V(F_i) \subseteq V(F_j)\cap\left(\bigcup_{h<j}V(F_h)\right).\]
    Since $\mathcal{F}$ is a forest of copies of $F$, the right-hand side is empty, a singleton, or the endpoints of an edge of $H$. 
    The assertion follows.

    % \ref{itm:forest-no-nonedge} Let $F_1 \neq F_2 \in \mathcal{F}$ and suppose that $u \neq v \in V(H)$ are nonadjacent in $H$ and $u, v \in V(F_1) \cap V(F_2)$.
    % Without loss of generality, we may assume that $F_2$ comes after $F_1$ in the forest ordering of $\mathcal{F}$.
    % When $F_2$ is added to the forest, the vertices $u$ and $v$ already belong to the union of the preceding copies, and they are nonadjacent in $H$, contradicting that $\mathcal{F}$ is a forest of copies of $F$ in $H$.
    
    \ref{itm:clique-common-member} Let $(T, \mathcal{V})$ be the standard tree-decomposition of $\Gamma(\mathcal{F})$.
    For every $x \in V(\Gamma(\mathcal{F}))$, the set $T_x = \{F_j \in \mathcal{F} : x \in V(F_j)\}$ induces a nonempty connected subtree of $T$.
    \cref{itm:clique-common-member} then follows from the Helly property of subtrees of a tree.
    
    \ref{itm:3-nbrs-common-member} Let $(T, \mathcal{V})$ be the standard tree-decomposition of $\Gamma(\mathcal{F})$.
    Consider $u, v \in V(\Gamma(\mathcal{F}))$ with at least three common neighbors.
    Suppose for a contradiction that the sets $T_u$ and $T_v$ are disjoint.
    Since each of them is connected, there is an edge $e$ of $T$ that separates them.
    Then, every common neighbor of $u$ and $v$ in $\Gamma(\mathcal{F})$ must lie in the adhesion set at $e$.
    This is a contradiction since the tree-decomposition has adhesion at most $2$ whereas $u$ and $v$ have at least three common neighbors in $\Gamma(\mathcal{F})$.
    Thus, $T_u \cap T_v \neq \emptyset$.
    Let $F_j \in T_u \cap T_v$. 
    By definition, we have $u, v \in V(F_j)$, so $u$ and $v$ occur in a common member of $\mathcal{F}$.
    Given a common neighbor $w$ of $u$ and $v$, the sets $T_u, T_v$ and $T_w$ intersect pairwise, so the Helly property of subtrees of a tree gives $F_j \in \mathcal{F}$ such that $F_j \in T_u \cap T_v \cap T_w$.
    Again, this means that $F_j$ contains all three of $u, v$ and $w$.
    
    \ref{itm:C4-common-member} Consider again the tree-decomposition $(T, \mathcal{V})$ of $\Gamma(\mathcal{F})$ defined above.
    Let $u_1u_2u_3u_4$ be a $4$-cycle in $\Gamma(\mathcal{F})$ such that $u_1u_3$ and $u_2u_4$ are non-edges of $H$.
    Suppose for a contradiction that the sets $T_{u_1}$ and $T_{u_3}$ are disjoint.
    Since each of them is connected, there is an edge $e$ of $T$ that separates them.
    Then, every common neighbor of $u_1$ and $u_3$ in $\Gamma(\mathcal{F})$ must lie in the adhesion corresponding to $e$, so this adhesion contains $u_2$ and $u_4$.
    Since every adhesion has size at most $2$, the adhesion at $e$ is $\{u_2, u_4\}$.
    This is a contradiction since $u_2u_4 \notin E(H)$ whereas every adhesion of size $2$ induces an edge of $H$.
    Thus, $T_{u_1} \cap T_{u_3} \neq \emptyset$.
    Similarly, we have $T_{u_2} \cap T_{u_4} \neq \emptyset$.
    Thus, the subtrees $T_{u_1}, T_{u_2}, T_{u_3}$ and $T_{u_4}$ intersect pairwise.
    By the Helly property of subtrees, they all contain some node $F_j \in V(T)$.
    Again, this means that $F_j$ contains $u_1, u_2, u_3$ and $u_4$.
\end{proof}

\subsection{The matrix construction}

In this section, we prove \cref{thm:matrix-not-disambig}.
We start by constructing a very tame partial matrix $U$ whose two uniform disambiguations have VC-dimension at least $k$.

\begin{lemma}\label{lem:universal-matrix-simple}
    For every $k \in \mathbb{N}$, there exists a partial matrix $U$ with the following properties: \begin{itemize}
        \item No row or column of $U$ contains both a $0$ entry and a $1$ entry.
        \item The matrix obtained by casting all the $\star$ entries of $U$ into $0$ entries (resp. $1$ entries) has VC-dimension at least $k$.
    \end{itemize}
\end{lemma}

\begin{proof}
    Let $S$ be the $2^k \times k$ matrix whose rows are indexed by the subsets $A \subseteq [k]$ and whose columns are indexed by the elements $i \in [k]$, such that $S[A, i] = 1$ if $i \in A$ and $S[A, i] = 0$ if $i \notin A$.
    The set of columns of $S$ is shattered, so $S$ has VC-dimension at least $k$.
    Let $S^+$ (resp. $S^-$) be the partial matrix obtained from $S$ by turning all $0$ entries (resp. $1$ entries) of $S$ into $\star$ entries.
    Let \[ U =
    \begin{bmatrix}
        S^+ & \star \\
        \star & S^-
    \end{bmatrix}.\]
    There are no $0$ entries in $S^+$ and no $1$ entries in $S^-$ so no row or column of $U$ contains both a $0$ entry and a $1$ entry.
    If all the $\star$ entries of $U$ are cast into $0$ entries (resp. $1$ entries), the resulting matrix contains $S$ as a submatrix, hence has VC-dimension at least $k$.
\end{proof}

We can now prove \cref{thm:matrix-not-disambig}, which we restate for convenience.

\mainmatrix*

\begin{proof}
    Let $k \in \mathbb{N}$ be given, and let $U$ be a partial matrix such that no row or column of $U$ contains both a $0$ entry and a $1$ entry, and such that for every $b \in \{0, 1\}$ the matrix $U^{(b)}$ obtained by casting all the $\star$ entries of $U$ into $b$ entries has VC-dimension at least $k$; such a matrix exists by \cref{lem:universal-matrix-simple}.
    
    Denote by $\mathcal{R}$ the set of rows of $U$ and by $\mathcal{C}$ the set of columns of $U$.
    Let $F$ be the bipartite graph with bipartition $(\mathcal{R}, \mathcal{C})$ such that $r \in \mathcal{R}$ is adjacent to $c \in \mathcal{C}$ if and only if $U[r, c] = \star$.
    Applying Girth Ramsey to the bipartite graph $F$ with prescribed bipartition $(\mathcal{R}, \mathcal{C})$, $r = 2$ and $\ell = 4$, we obtain a bipartite graph $H$ with prescribed bipartition $(\widehat{\mathcal{R}}, \widehat{\mathcal{C}})$ and a set $\mathcal{F}$ of induced copies of $F$ in $H$ with the following properties: \begin{itemize}
        \item Every family $\mathcal{S} \subseteq \mathcal{F}$ of size at most $4$ is contained in a forest of copies of $F$.
        \item For every coloring $\chi : E(H) \to \{0, 1\}$, some member of $\mathcal{F}$ is monochromatic.
    \end{itemize}
    For every $Q \in \mathcal{F}$, fix an isomorphism $\phi_Q : V(F) \to V(Q)$ that maps $\mathcal{R}$ inside $\widehat{\mathcal{R}}$ and $\mathcal{C}$ inside $\widehat{\mathcal{C}}$; such an isomorphism exists by the `Moreover' part of \cref{thm:girth-ramsey}.
    
    Let $\widehat{u}\widehat{v} \notin E(H)$. 
    We show that there is at most one copy $Q \in \mathcal{F}$ that contains both $\widehat{u}$ and $\widehat{v}$.
    Suppose for a contradiction that there are two distinct such copies $Q, Q' \in \mathcal{F}$.
    Since $\ell \geq 2$, there is a forest $\mathcal{F}' \subseteq \mathcal{F}$ of copies of $F$ that contains both $Q$ and $Q'$.
    By \cref{lem:properties-forest}~\ref{itm:forest-no-nonedge} applied to $\mathcal{F'}$, $Q$ and $Q'$ cannot share vertices that are nonadjacent in $H$, which contradicts that they share $\widehat{u}$ and $\widehat{v}$.
    
    We now define a partial matrix $M$ with row set $\widehat{\mathcal{R}}$ and column set $\widehat{\mathcal{C}}$, as follows. 
    For $\widehat{r} \in \widehat{\mathcal{R}}$ and $\widehat{c} \in \widehat{\mathcal{C}}$, set
    \[M[\widehat{r}, \widehat{c}] = \begin{cases}
      \star &\text{if $\widehat{r}\widehat{c} \in E(H)$ or no $Q \in \mathcal{F}$ contains both $\widehat{r}$ and $\widehat{c}$,} \\
      U[\phi_Q^{-1}(\widehat{r}), \phi_Q^{-1}(\widehat{c})] &\text{if $\widehat{r}\widehat{c} \notin E(H)$ and $Q \in \mathcal{F}$ contains both $\widehat{r}$ and $\widehat{c}$.}
    \end{cases}\]
    This matrix is well-defined by the previous observation and since $\phi_Q$ maps $\mathcal{R}$ inside $\widehat{\mathcal{R}}$ and $\mathcal{C}$ inside $\widehat{\mathcal{C}}$.
    Furthermore, for every $Q \in \mathcal{F}$, we have $M[\phi_Q(\mathcal{R}), \phi_Q(\mathcal{C})] = U$, if the columns and the rows of $M[\phi_Q(\mathcal{R}), \phi_Q(\mathcal{C})]$ are ordered according to the preimage under $\phi_Q$. This follows from the definition and from the fact that edges of $H$ correspond to $\star$-entries in $M$, and also correspond to edges of $F$, hence $\star$-entries of $U$.
    Thus, every $Q \in \mathcal{F}$ is also a copy of $U$ in $M$.
    
    We now show that $M$ has the desired properties.
    First, let $M'$ be a disambiguation of $M$.
    Then, $M'$ naturally corresponds to a $2$-coloring $\chi$ of the $\star$ entries of $M$ with colors $0$ and $1$.
    Recall that $H$ is bipartite with bipartition $(\widehat{\mathcal{R}}, \widehat{\mathcal{C}})$, so every edge of $H$ corresponds to a $\star$ entry of $M$ by definition of $M$ so $\chi$ restricts to a coloring $E(H) \to \{0, 1\}$, which we still call $\chi$.
    By construction of $H$, some member $Q \in \mathcal{F}$ is monochromatic for $\chi$, with color $b \in \{0, 1\}$.
    We saw that $M[\phi_Q(\mathcal{R}),\phi_Q(\mathcal{C})]$ is isomorphic to $U$.
    Since $Q \in \mathcal{F}$ is monochromatic for $\chi$, this means that in the copy $M[\phi_Q(\mathcal{R}),\phi_Q(\mathcal{C})]$ of $U$ in $M$, all the $\star$ entries were cast to $b$.
    Thus, $M'[\phi_Q(\mathcal{R}),\phi_Q(\mathcal{C})]$ is isomorphic (up to a permutation of the rows and of the columns) to $U^{(b)}$.
    By the construction of $U$, $U^{(b)}$ has VC-dimension at least $k$.
    Importantly, this holds regardless of the permutation of the rows and of the columns.
    
    Finally, we argue that every pure $2 \times 2$ submatrix of $M$ is constant.
    Consider distinct rows $r_1, r_2 \in \widehat{\mathcal{R}}$ and distinct columns $c_1, c_2 \in \widehat{\mathcal{C}}$ such that the submatrix $M[\{r_1, r_2\}, \{c_1, c_2\}]$ is pure.
    Each of the four corresponding entries is witnessed by a member $Q$ of $\mathcal{F}$.
    By assumption, since there are at most $4$ such members, there is a forest $\mathcal{F}' \subseteq \mathcal{F}$ of copies of $F$ in $H$ that contains all of them.
    Then, $r_1, c_1, r_2, c_2$ form a $4$-cycle in $\Gamma(\mathcal{F}')$.
    Furthermore, $r_1r_2$ and $c_1c_2$ are non-edges of $H$ since $r_1, r_2 \in \widehat{\mathcal{R}}$ and $c_1, c_2 \in \widehat{\mathcal{C}}$ and $H$ is bipartite with bipartition $(\widehat{\mathcal{R}}, \widehat{\mathcal{C}})$.
    By \cref{lem:properties-forest}~\ref{itm:C4-common-member}, some member $Q$ of $\mathcal{F}'$ contains $r_1, r_2, c_1$ and $c_2$.
    By construction of $M$, we have $M[\{r_1, r_2\}, \{c_1, c_2\}] = U[\phi_Q^{-1}(\{r_1, r_2\}), \phi_Q^{-1}(\{c_1, c_2\})]$.
    Therefore, the matrix $U[\phi_Q^{-1}(\{r_1, r_2\}), \phi_Q^{-1}(\{c_1, c_2\})]$ is pure.
    However, no row or column of $U$ contains both a $0$ entry and a $1$ entry, so the matrix $U[\phi_Q^{-1}(\{r_1, r_2\}), \phi_Q^{-1}(\{c_1, c_2\})]$ is constant.
    Thus, $M[\{r_1, r_2\}, \{c_1, c_2\}]$ is constant too.
\end{proof}

\subsection{The graph construction}

In this section, we prove \cref{thm:graph-not-disambig}.
The proof closely follows the proof of \cref{thm:matrix-not-disambig}.

A \emph{partial graph} $G$ is a pair $(V, \mathcal{E})$ where $\mathcal{E}$ is a partition of $\binom{V}{2}$ into three (potentially empty) sets $E_1$, $E_\star$ and $E_0$.
Given a partial graph $G$, we denote the corresponding sets by $E_1(G), E_{\star}(G)$ and $E_0(G)$.
We refer to the elements of $E_1(G)$ as the \emph{$1$-edges}, or simply as the \emph{edges}, to the elements of $E_{\star}(G)$ as the \emph{$\star$-edges}, and to the elements of $E_0(G)$ as the \emph{$0$-edges}, or simply the \emph{non-edges}.
For every $v \in V(G)$, we set $N_1(v) = \{u \in V(G) : uv \in E_1(G)\}$, $N_{\star}(v) = \{v\} \cup \{u \in V(G) : uv \in E_\star(G)\}$ and $N_0(v) = \{u \in V(G) : uv \in E_0(G)\}$.
Each graph $G = (V, E)$ can be viewed as a partial graph by setting $E_1(G) = E(G)$, $E_\star(G) = \emptyset$ and $E_0(G) = \binom{V(G)}{2} \setminus E_1(G)$.
Recall that the VC-dimension of a partial graph is that of its adjacency matrix (with $\star$ on the diagonal).

A \emph{disambiguation} of a partial graph $G$ is a graph $G'$ such that $V(G') = V(G)$, $E_1(G) \subseteq E(G')$ and $E_0(G) \cap E(G') = \emptyset$.
In other words, $G'$ is a disambiguation of $G$ when $G'$ can be obtained from $G$ by specifying for each $\star$-edge of $G$ whether it should become an edge or a non-edge in $G'$.

Again, we start by constructing a very tame partial graph $U$ whose two uniform disambiguations contain all $k$-vertex graphs as induced subgraphs.

\begin{lemma}\label{lem:universal-graph-simple}
    For every $k \in \mathbb{N}$, there exists a partial graph $U$ with the following properties: \begin{itemize}
        \item No vertex of $U$ is incident to both an edge and a non-edge.
        \item The graph obtained by casting all the $\star$s in $U$ into edges (resp. non-edges) contains all $k$-vertex graphs as induced subgraphs.
    \end{itemize}
\end{lemma}

\begin{proof}
    Let $D$ be a graph that contains all $k$-vertex graphs as induced subgraphs (up to isomorphism); the disjoint union of all $k$-vertex graphs (up to isomorphism) is such a graph.
    Let $D^+$ (resp. $D^-$) be the partial graph obtained from $D$ by turning all non-edges (resp. edges) of $D$ into $\star$s.
    Let $U$ be obtained from a copy of $D^+$ and a copy of $D^-$ by adding a $\star$ between any vertex of $D^+$ and any vertex of $D^-$.
    There are no edges in $D^-$ and no non-edges in $D^+$ and we only added $\star$s to construct $U$ from $D^+$ and $D^-$ so no vertex of $U$ is incident to both an edge and a non-edge.
    If all the $\star$s of $U$ are cast into edges (resp. non-edges), the resulting graph contains an induced copy of $D$ on vertex set $V(D^-)$ (resp. $V(D^+)$), hence contains all $k$-vertex graphs as induced subgraphs.
\end{proof}

Note that the first property implies that every pure induced subgraph of $U$ is either a clique or a stable set.
We can now prove \cref{thm:graph-not-disambig}, which we restate for convenience.

\maingraph*

\begin{proof}
    Let $k \in \mathbb{N}$ be given, and let $U$ be a partial graph such that no vertex of $U$ is incident to both an edge and a non-edge, and such that the graph $U^{(1)}$ (resp. $U^{(0)}$) obtained by casting all the $\star$s in $U$ into edges (resp. non-edges) contains all $k$-vertex graphs as induced subgraphs; such a graph exists by \cref{lem:universal-graph-simple}.
    
    Let $F$ be the graph on vertex set $V(U)$ whose edges correspond exactly to the $\star$s in $U$: for all $u, v \in V(F)$, $uv \in E(F) \iff uv \in E_\star(U)$.
    Applying Girth Ramsey to $F$ with $r = 2$ and $\ell = 8$, we obtain a graph $H$ and a set $\mathcal{F}$ of induced copies of $F$ in $H$ with the following properties: \begin{itemize}
        \item Every family $\mathcal{S} \subseteq \mathcal{F}$ of size at most $8$ is contained in a forest of copies of $F$.
        \item For every coloring $\chi : E(H) \to \{0, 1\}$, some member of $\mathcal{F}$ is monochromatic.
    \end{itemize}
    For every $Q \in \mathcal{F}$, fix an isomorphism $\phi_Q : V(F) \to V(Q)$.
    
    Let $uv \notin E(H)$. 
    We show that there is at most one copy $Q \in \mathcal{F}$ that contains both $u$ and $v$.
    Suppose for a contradiction that there are two distinct such copies $Q, Q' \in \mathcal{F}$.
    Since $\ell \geq 2$, there is a forest $\mathcal{F}' \subseteq \mathcal{F}$ of copies of $F$ that contains both $Q$ and $Q'$.
    By \cref{lem:properties-forest}~\ref{itm:forest-no-nonedge} applied to $\mathcal{F'}$, $Q$ and $Q'$ cannot share vertices that are nonadjacent in $H$, which contradicts that they share $u$ and $v$.
    
    We now define a partial graph $G$ on vertex set $V(H)$, as follows. For $u \neq v \in V(H)$, set
    \[G(u, v) = \begin{cases}
      \star &\text{if $uv \in E(H)$ or no $Q \in \mathcal{F}$ contains both $u$ and $v$,} \\
      U(\phi_Q^{-1}(u), \phi_Q^{-1}(v)) &\text{if $uv \notin E(H)$ and $Q \in \mathcal{F}$ contains both $u$ and $v$.}
    \end{cases}\]
    This graph is well-defined by the previous observation.
    Furthermore, for every $Q \in \mathcal{F}$, $\phi_Q : V(U) = V(F) \to V(Q)$ is an isomorphism between the partial graphs $U$ and $G[V(Q)]$. This follows from the definition and from the fact that edges of $H$ correspond to $\star$-edges in $G$, and also correspond to edges of $F$, hence $\star$-edges of $U$.
    Thus, every $Q \in \mathcal{F}$ is also a copy of $U$ in $G$.
    
    We now show that $G$ has the desired properties.
    First, let $G'$ be a disambiguation of $G$.
    Then, $G'$ naturally corresponds to a $2$-coloring $\chi : E_\star(G) \to \{0, 1\}$.
    By definition of $G$, we have $E(H) \subseteq E_\star(G)$ so $\chi$ restricts to a coloring $E(H) \to \{0, 1\}$, which we still call $\chi$.
    By construction of $H$, some member $Q \in \mathcal{F}$ is monochromatic for $\chi$, with color $b \in \{0, 1\}$.
    This means that in the copy $Q$ of $U$ in $G$, all the $\star$-edges were cast according to $b$.
    Thus, $G'[V(Q)]$ is isomorphic to the graph $U^{(b)}$.
    By the construction of $U$, $U^{(b)}$ contains all $k$-vertex graphs as induced subgraphs, so $G'$ contains all $k$-vertex graphs as induced subgraphs.
    
    We now argue that every pure induced subgraph of $G$ on $3$ vertices is homogeneous.
    Let $u, v, w$ form a pure induced subgraph of $G$ on $3$ vertices.
    By construction of $G$, any two vertices in $\{u, v, w\}$ lie in a common member of $\mathcal{F}$.
    Denote these members of $\mathcal{F}$ by $Q_{u, v}, Q_{u, w}, Q_{v, w}$.
    By assumption, there is a forest $\mathcal{F}' \subseteq \mathcal{F}$ of copies of $F$ in $H$ that contains $Q_{u, v}, Q_{u, w}$, and $Q_{v, w}$.
    Thus, $u, v$ and $w$ form a clique in the graph $\Gamma(\mathcal{F}')$, so by \cref{lem:properties-forest}~\ref{itm:clique-common-member}, some $Q \in \mathcal{F}'$ contains all three of $u, v, w$.
    By construction of $G$, we have $G[\{u, v, w\}] \cong U[\phi_Q^{-1}(\{u, v, w\})]$.
    However, no vertex of $U$ is incident to both an edge and a non-edge.
    Since $G[\{u, v, w\}]$ is pure, this implies that $G[\{u, v, w\}]$ is homogeneous.
    
    Finally, we argue that $G$ has VC-dimension at most $1$.
    Suppose for a contradiction that there exist vertices $u, v \in V(G)$ and vertices $x_{00},x_{01},x_{10},x_{11} \in V(G)$ such that \begin{align*}
        &ux_{10}, ux_{11}, vx_{01}, vx_{11} \in E_1(G), \text{ and } \\
        &ux_{00}, ux_{01}, vx_{00}, vx_{10} \in E_0(G).
    \end{align*}
    First, observe that the sets $\{u, v\}$ and $\{x_{00}, x_{01}, x_{10}, x_{11}\}$ are disjoint since $uu, vv \notin E_0(G) \cup E_1(G)$.
    Each of these edges and non-edges is witnessed by a member $Q$ of $\mathcal{F}$.
    By assumption, since there are at most $8$ such members, there is a forest $\mathcal{F}' \subseteq \mathcal{F}$ of copies of $F$ in $H$ that contains all of them.
    Thus, $x_{00},x_{01},x_{10},x_{11}$ are pairwise distinct common neighbors of $u$ and $v$ in $\Gamma(\mathcal{F}')$.
    By \cref{lem:properties-forest}~\ref{itm:3-nbrs-common-member}, some member $Q$ of $\mathcal{F}'$ contains $u, v, x_{01}$.
    By construction of $G$, we have $G[\{u, v, x_{01}\}] \cong U[\phi_Q^{-1}(\{u, v, x_{01}\})]$.
    Therefore, $\phi_Q^{-1}(x_{01})$ is incident to both an edge and a non-edge in $U$.
    This contradicts the definition of $U$.
    Thus, $G$ has no shattered set of size $2$, so has VC-dimension at most $1$.
\end{proof}

\bibliographystyle{alphaurl}
\bibliography{biblio.bib}

@article{A83,
  title = {Densit{\'e} et dimension},
  author = {Patrice Assouad},
  journal = {Annales de l'Institut Fourier},
  year = {1983},
  volume = {33},
  pages = {233-282},
  doi = {10.5802/aif.938},
}

@article{AKM22,
  author  = {Idan Attias and Aryeh Kontorovich and Yishay Mansour},
  title   = {Improved Generalization Bounds for Adversarially Robust Learning},
  journal = {Journal of Machine Learning Research},
  year    = {2022},
  volume  = {23},
  number  = {175},
  pages   = {1--31},
  url     = {http://jmlr.org/papers/v23/20-1353.html},
  eprint={1810.02180},
  archivePrefix={arXiv},
}

@misc{Palvolgyi26,
      title={A nearcircumsphere-{R}amsey Theorem for Solvable Transitive Configurations}, 
      author={Dömötör Pálvölgyi},
      year={2026},
      eprint={2608.10865},
      archivePrefix={arXiv},
      primaryClass={math.CO},
}

@incollection{KST17,
  author    = {Kechris, Alexander S. and Soki{\'c}, Miodrag and Todor{\v{c}}evi{\'c}, Stevo},
  title     = {Ramsey properties of finite measure algebras and topological dynamics of the group of measure preserving automorphisms: some results and an open problem},
  booktitle = {Modus Operandi: Essays in honor of Hugh Woodin},
  series    = {Contemporary Mathematics},
  volume    = {690},
  pages     = {215--249},
  publisher = {American Mathematical Society},
  address   = {Providence, RI},
  year      = {2017},
  note      = {Dedicated to Hugh Woodin on his 60th birthday},
  DOI = {10.1090/conm/690}
}

@article{Shelah72,
  author = {Saharon Shelah},
  title = {{A combinatorial problem; stability and order for models and theories in infinitary languages.}},
  volume = {41},
  journal = {Pacific Journal of Mathematics},
  number = {1},
  publisher = {Pacific Journal of Mathematics, A Non-profit Corporation},
  pages = {247--261},
  year = {1972},
}

@article{Sauer72,
  author = {Norbert Sauer},
  title = {On the density of families of sets},
  journal = {Journal of Combinatorial Theory, Series A},
  year = {1972},
  volume = {13},
  number = {1},
  pages = {145-147},
  issn = {0097-3165},
  doi = {https://doi.org/10.1016/0097-3165(72)90019-2},
}

@INPROCEEDINGS{AHHM22,
  author={Alon, Noga and Hanneke, Steve and Holzman, Ron and Moran, Shay},
  booktitle={2021 IEEE 62nd Annual Symposium on Foundations of Computer Science (FOCS)}, 
  title={A Theory of {PAC} Learnability of Partial Concept Classes}, 
  year={2022},
  volume={},
  number={},
  pages={658-671},
  eprint={2107.08444},
  archivePrefix={arXiv},
  primaryClass={math.CO},
  doi={10.1109/FOCS52979.2021.00070}}

@article{BBGJK21,
author = {Balodis, Kaspars and Ben-David, Shalev and G{\"o}{\"o}s, Mika and Jain, Siddhartha and Kothari, Robin},
title = {Unambiguous {DNF}s and {A}lon–{S}aks–{S}eymour},
journal = {SIAM Journal on Computing},
volume = {0},
number = {0},
pages = {FOCS21-157-FOCS21-173},
year = {2021},
doi = {10.1137/22M1480616},
eprint = {2102.08348},
archivePrefix={arXiv}
}

@inproceedings{BHHLT26,
author = {Blondal, Ari and Hatami, Hamed and Hatami, Pooya and Lalov, Chavdar and Tretiak, Sivan},
title = {Borsuk-{U}lam and Replicable Learning of Large-Margin Halfspaces},
year = {2026},
isbn = {9798400725364},
publisher = {Association for Computing Machinery},
address = {New York, NY, USA},
doi = {10.1145/3798129.3800771},
booktitle = {Proceedings of the 58th Annual ACM Symposium on Theory of Computing},
pages = {529–540},
numpages = {12},
location = {Salt Lake City, UT, USA},
series = {STOC '26},
eprint={2503.15294},
archivePrefix={arXiv}
}

@InProceedings{CHHH23,
  author =	{Cheung, Tsun-Ming and Hatami, Hamed and Hatami, Pooya and Hosseini, Kaave},
  title =	{Online Learning and Disambiguations of Partial Concept Classes},
  booktitle =	{50th International Colloquium on Automata, Languages, and Programming (ICALP 2023)},
  pages =	{42:1--42:13},
  series =	{Leibniz International Proceedings in Informatics (LIPIcs)},
  ISBN =	{978-3-95977-278-5},
  ISSN =	{1868-8969},
  year =	{2023},
  volume =	{261},
  editor =	{Etessami, Kousha and Feige, Uriel and Puppis, Gabriele},
  publisher =	{Schloss Dagstuhl -- Leibniz-Zentrum f{\"u}r Informatik},
  address =	{Dagstuhl, Germany},
  URN =		{urn:nbn:de:0030-drops-180946},
  doi =		{10.4230/LIPIcs.ICALP.2023.42},
  eprint={2303.17578},
  archivePrefix={arXiv}
}

@misc{CMW25,
      title={Spherical dimension}, 
      author={Bogdan Chornomaz and Shay Moran and Tom Waknine},
      year={2025},
      eprint={2503.10240},
      archivePrefix={arXiv},
      primaryClass={cs.DM},
}

@INPROCEEDINGS{FHMST24,
  author={Fioravanti, Simone and Hanneke, Steve and Moran, Shay and Schefler, Hilla and Tsubari, Iska},
  booktitle={2024 IEEE 65th Annual Symposium on Foundations of Computer Science (FOCS)}, 
  title={Ramsey Theorems for Trees and a General ‘Private Learning Implies Online Learning’ Theorem}, 
  year={2024},
  volume={},
  number={},
  pages={1983-2009},
  doi={10.1109/FOCS61266.2024.00119},
  eprint={2407.07765},
  archivePrefix={arXiv},
  }

@misc{IJN26,
      title={The {VC} dimension of partial concept classes via {R}adon's theorem}, 
      author={Grigory Ivanov and Attila Jung and Márton Naszódi},
      year={2026},
      eprint={2607.10751},
      archivePrefix={arXiv},
      primaryClass={cs.LG},
}

@misc{Pabbaraju26,
      title={Optimal Unambiguous {DNF}s and {A}lon-{S}aks-{S}eymour}, 
      author={Chirag Pabbaraju},
      year={2026},
      eprint={2608.02533},
      archivePrefix={arXiv},
      primaryClass={cs.CC},
}

@INPROCEEDINGS{Goos15,
  author={Göös, Mika},
  booktitle={2015 IEEE 56th Annual Symposium on Foundations of Computer Science}, 
  title={Lower Bounds for Clique vs. Independent Set}, 
  year={2015},
  volume={},
  number={},
  pages={1066-1076},
  doi={10.1109/FOCS.2015.69}}

@misc{RR23,
      title={The girth {R}amsey theorem}, 
      author={Christian Reiher and Vojtěch Rödl},
      year={2023},
      eprint={2308.15589},
      archivePrefix={arXiv},
      primaryClass={math.CO},
}

@article{HK26,
   title={Twenty years of {N}ešetřil’s classification programme of {R}amsey classes},
   volume={59},
   ISSN={1574-0137},
   DOI={10.1016/j.cosrev.2025.100814},
   journal={Computer Science Review},
   publisher={Elsevier BV},
   author={Hubička, Jan and Konečný, Matěj},
   year={2026},
   month=Feb, pages={100814},
   eprint={2501.17293},
   archivePrefix={arXiv},
}

@article{GR71,
 ISSN = {00029947},
 author = {Ronald. L. Graham and Bruce. L. Rothschild},
 journal = {Transactions of the American Mathematical Society},
 pages = {257--292},
 publisher = {American Mathematical Society},
 title = {Ramsey's Theorem for $n$-Parameter Sets},
 urldate = {2026-09-08},
 volume = {159},
 year = {1971},
 DOI = {10.2307/1996010}
}

@article{FR90,
 ISSN = {08940347, 10886834},
 author = {Péter Frankl and Vojtěch Rödl},
 journal = {Journal of the American Mathematical Society},
 number = {1},
 pages = {1--7},
 publisher = {American Mathematical Society},
 title = {A Partition Property of Simplices in {E}uclidean Space},
 urldate = {2026-09-08},
 volume = {3},
 year = {1990},
 DOI = {10.2307/1990982}
}

@INPROCEEDINGS{BCT25,
  author={Bourneuf, Romain and Charbit, Pierre and Thomassé, Stéphan},
  booktitle={2025 IEEE 66th Annual Symposium on Foundations of Computer Science (FOCS)}, 
  title={A Dense Neighborhood Lemma: Applications of Partial Concept Classes to Domination and Chromatic Number}, 
  year={2025},
  volume={},
  number={},
  pages={1-37},
  doi={10.1109/FOCS63196.2025.00007},
  eprint={2504.02992},
  archivePrefix={arXiv},
  }

@Article{NR81,
author={Ne{\v{s}}et{\v{r}}il, Jaroslav
and R{\"o}dl, Vojt{\v{e}}ch},
title={Simple proof of the existence of restricted {R}amsey graphs by means of a partite construction},
journal={Combinatorica},
year={1981},
month={Jun},
day={01},
volume={1},
number={2},
pages={199-202},
issn={1439-6912},
doi={10.1007/BF02579274},
}

@inproceedings{L01,
author = {Long, Philip M.},
title = {On Agnostic Learning with \{0, *, 1\}-Valued and Real-Valued Hypotheses},
year = {2001},
isbn = {3540423435},
publisher = {Springer-Verlag},
address = {Berlin, Heidelberg},
booktitle = {Proceedings of the 14th Annual Conference on Computational Learning Theory and and 5th European Conference on Computational Learning Theory},
pages = {289–302},
numpages = {14},
series = {COLT '01/EuroCOLT '01},
DOI = {10.1007/3-540-44581-1_19}
}

@inproceedings{HHM23,
author = {Hatami, Hamed and Hosseini, Kaave and Meng, Xiang},
title = {A {B}orsuk-{U}lam Lower Bound for Sign-Rank and Its Applications},
year = {2023},
isbn = {9781450399135},
publisher = {Association for Computing Machinery},
address = {New York, NY, USA},
doi = {10.1145/3564246.3585210},
booktitle = {Proceedings of the 55th Annual ACM Symposium on Theory of Computing},
pages = {463–471},
numpages = {9},
location = {Orlando, FL, USA},
series = {STOC 2023}
}

@misc{FHV26,
      title={A $\mathbb{Z}_2$-Topological Framework for Sign-rank Lower Bounds}, 
      author={Florian Frick and Kaave Hosseini and Aliaksei Vasileuski},
      year={2026},
      eprint={2604.01510},
      archivePrefix={arXiv},
      primaryClass={math.CO},
}

@incollection{Deuber75,
  author    = {Deuber, Walter},
  title     = {A generalization of {Ramsey}'s theorem},
  booktitle = {Infinite and Finite Sets, Vol. I},
  series    = {Colloquia Mathematica Societatis J{\'a}nos Bolyai},
  volume    = {10},
  publisher = {North-Holland},
  address   = {Amsterdam/London},
  year      = {1975},
  pages     = {323--332}
}

@incollection{EHP75,
  author    = {Erd{\H{o}}s, Paul and Hajnal, András and P{\'o}sa, Lajos},
  title     = {Strong embeddings of graphs into colored graphs},
  booktitle = {Infinite and Finite Sets, Vol. I},
  series    = {Colloquia Mathematica Societatis J{\'a}nos Bolyai},
  volume    = {10},
  publisher = {North-Holland},
  address   = {Amsterdam/London},
  year      = {1975},
  pages     = {585--595}
}

@mastersthesis{Rodl73,
  author = {R{\"o}dl, Vojtěch},
  title  = {The dimension of a graph and generalized {Ramsey} theorems},
  school = {Charles University},
  year   = {1973}
}

@article{Rosenblatt58,
  author  = {Rosenblatt, Frank},
  title   = {The perceptron: A probabilistic model for information
             storage and organization in the brain},
  journal = {Psychological Review},
  volume  = {65},
  number  = {6},
  pages   = {386--408},
  year    = {1958},
  doi     = {10.1037/h0042519}
}

@inproceedings{Novikoff63,
  author    = {Novikoff, Albert B. J.},
  title     = {On convergence proofs for perceptrons},
  booktitle = {Proceedings of the Symposium on Mathematical Theory
               of Automata},
  editor    = {Fox, Jerome},
  series    = {Microwave Research Institute Symposia Series},
  volume    = {12},
  pages     = {615--622},
  publisher = {Polytechnic Press},
  address   = {Brooklyn, NY},
  year      = {1963},
}

\appendix

\section{The radial Ramsey approach}\label{sec:hdr}

The unconditional proof of \cref{thm:main-cube} obtains distance-profile-canonicality using two potentially different embeddings for the rows and columns.
The radial approach seeks a stronger restriction for symmetric colorings: a single near-spanning homothetic copy on which the color depends only on total distance.
In this appendix, we develop the conditional approach outlined in the introduction.
We first show that the Near-Spanning Radial Ramsey Conjecture implies the Gap-Hamming Disambiguation theorem.
We then derive the radial conjecture from the Homogeneous Dual Ramsey Conjecture, through a Ramsey statement about partitions with prescribed part sizes.

We begin with the terminology used in \cref{conj:radial}.
Let $(X,d_X)$ and $(Y,d_Y)$ be metric spaces.
A map $\phi:X\to Y$ is a \emph{homothety} if there exists a constant $\lambda>0$, called its \emph{ratio}, such that
\[
    d_Y(\phi(x),\phi(x'))=\lambda d_X(x,x')
    \qquad\text{for all }x,x'\in X.
\]
We call $\phi$ \emph{$\varepsilon$-near-spanning} if
\[
    \diam_{d_Y}(\phi(X))
    \geq (1-\varepsilon)\diam_{d_Y}(Y).
\]
Given a coloring $c:Y\times Y\to\{0,1\}$, we call $\phi$
\emph{$c$-radial} if $c(\phi(x),\phi(x'))$ depends only on
$d_X(x,x')$.

We can now restate the Near-Spanning Radial Ramsey Conjecture.

\nearspanningradialramsey*

The symmetry assumption is necessary.
Indeed, fix a total order $\prec$ on $Q_n$ and set $c(x,y)=1$
exactly when $x\prec y$.
For any two distinct vertices $u,v$, we then have
$c(u,v)\neq c(v,u)$, whereas radiality requires these two
colors to agree.

The first step is to establish the following implication.

\begin{restatable}{lemma}{radialimpliesghs}\label{lem:radial-implies-ghs}
    The Near-Spanning Radial Ramsey Conjecture implies the
    Gap-Hamming Disambiguation theorem.
\end{restatable}

We then connect the radial conjecture to homogeneous dual Ramsey theory.
Ordinary Ramsey theory concerns subsets and inclusion, while dual Ramsey theory concerns partitions and coarsening.
Viewing subsets as injections and partitions as surjections makes this reversal explicit.
The so-called Dual Ramsey theorem was proved by Graham and Rothschild \cite{GR71}.
The Homogeneous Dual Ramsey Conjecture imposes the additional
requirement that the parts have equal size.
It was first conjectured by Kechris, Soki{\'c}, and Todor{\v{c}}evi{\'c} \cite{KST17}, and was still recorded as an open statement in a recent survey by Hubička and Konečný \cite{HK26}.
Recall that a $k$-equipartition is a partition into $k$ parts of equal size, and that a coarsening is obtained by merging parts.

\homogeneousdualramsey*

Our main result in this appendix is the following.

\begin{theorem}\label{thm:hdr-implies-radial}
    The Homogeneous Dual Ramsey Conjecture implies the
    Near-Spanning Radial Ramsey Conjecture.
\end{theorem}

Together with \cref{lem:radial-implies-ghs}, this gives the
conditional route to our main theorem.

\begin{corollary}\label{thm:hdr-implies-main}
    The Homogeneous Dual Ramsey Conjecture implies the
    Gap-Hamming Disambiguation theorem.
\end{corollary}

\begin{remark}[Comparison with vertex colorings]
For vertex colorings, monochromatic near-spanning homothetic copies of cubes are already guaranteed.
For every prescribed dimension and every $\varepsilon>0$, such copies exist in every two-coloring of a sufficiently large cube.
This follows from Pálvölgyi's Dense Block theorem (\cref{lem:palvolgyi}), applied to the natural action of
$C_2$ on $\{0,1\}$, and can also be derived from earlier work of Frankl and Rödl~\cite{FR90}.

Both the scaling and the near-spanning allowance are necessary. Coloring vertices by the parity of their Hamming weight excludes monochromatic isometric copies of $Q_1$. 
Coloring them by their first coordinate excludes monochromatic spanning copies of $Q_1$.

For pair colorings, even a near-spanning homothetic copy cannot generally be monochromatic.
Indeed, color pairs in $Q_n$ according to whether their distance is at most $n/2$ or greater than $n/2$.
Every $\varepsilon$-near-spanning homothetic copy of $Q_2$ then contains pairs of both colors.
This motivates seeking radiality rather than monochromaticity.
\end{remark}

\subsection{Near-Spanning Radial Ramsey implies Gap-Hamming Disambiguation}

We first show how to make a gap coloring symmetric while keeping control over its VC-dimension.
The argument uses a bound on the VC-dimension of unions, which we formulate in the language of set systems.

A \emph{set system} on a ground set $V$ is a collection $\mathcal S$ of subsets of $V$.
Its VC-dimension is that of the binary matrix $M$ whose rows are indexed by $\mathcal S$, whose columns are indexed by $V$, and whose entries satisfy $M[S,v]=1$ if $v\in S$ and $M[S,v]=0$ otherwise.

The following bound of Attias, Kontorovich, and Mansour~\cite{AKM22} follows from the Sauer--Shelah Lemma~\cite{Sauer72,Shelah72}.

\begin{lemma}[{\cite[Lemma 16]{AKM22}}] \label{lem:vcdim-union}
    Let $\mathcal{S}_1, \mathcal{S}_2$ be two set systems on the same ground set, with VC-dimensions $d_1$ and $d_2$, respectively.
    Then the set system $\{A_1 \cup A_2 : A_1 \in \mathcal{S}_1, A_2 \in \mathcal{S}_2\}$ has VC-dimension at most $2\log_2(6)(d_1+d_2)$.
\end{lemma}

We now use this bound to control the effect of symmetrization on VC-dimension.

\begin{lemma}\label{lem:make-c-symmetric}
    Let $0 < \varepsilon < 1/2$ and $n, q \in \mathbb{N}$.
    For every $\varepsilon$-gap coloring $c: Q_n \times Q_n \to \{0, 1\}$ of VC-dimension at most $q$, there exists a symmetric $\varepsilon$-gap coloring $\widetilde{c} : Q_n \times Q_n \to \{0, 1\}$ of VC-dimension at most $B(q) \coloneqq 2\log_2(6)(q+2^{q+1}-1)$.
\end{lemma}

\begin{proof}
    For all $x, y \in Q_n$, set $\widetilde{c}(x, y) = c(x, y) \vee c(y, x)$.
    Then, $\widetilde{c} : Q_n \times Q_n \to \{0, 1\}$ is a symmetric coloring.
    Furthermore, if $d(x, y) \leq \varepsilon \cdot \diam(Q_n)$ then $c(x, y) = 0 = c(y, x)$ so $\widetilde{c}(x, y) = 0$.
    Similarly, if $d(x, y) \geq (1-\varepsilon) \cdot \diam(Q_n)$ then $c(x, y) = 1 = c(y, x)$ so $\widetilde{c}(x, y) = 1$.
    Therefore, $\widetilde{c}$ is an $\varepsilon$-gap coloring.
    
    For every $u \in Q_n$, let $R_u = \{v \in Q_n : c(u, v) = 1\}$ and $C_u = \{v \in Q_n : c(v, u) = 1\}$.
    Consider the two set systems $\mathcal{R} = (Q_n, \{R_u : u \in Q_n\})$ and $\mathcal{C} = (Q_n, \{C_u : u \in Q_n\})$.
    Up to repetitions, $\mathcal{C}$ is the dual set system of $\mathcal{R}$.\footnote{Repeated rows of the original matrix give indistinguishable ground elements in the dual, while repeated columns give repeated sets. Removing either kind of repetition does not affect VC-dimension.}
    By definition, the VC-dimension of $\mathcal{R}$ is the VC-dimension of $c$, so $\mathcal{R}$ has VC-dimension at most $q$.
    By Assouad's bound \cite{A83}, it follows that $\mathcal{C}$ has VC-dimension at most $2^{q+1}-1$.
    The VC-dimension of $\widetilde{c}$ is the VC-dimension of the set system $\mathcal{S} = (Q_n, \{R_{u} \cup C_{u} : u \in Q_n\})$.
    Moreover, $\{R_{u} \cup C_{u} : u \in Q_n\} \subseteq \{R_{u_1} \cup C_{u_2} : u_1, u_2 \in Q_n\}$.
    By \cref{lem:vcdim-union}, the set system $(Q_n, \{R_{u_1} \cup C_{u_2} : u_1, u_2 \in Q_n\})$ has VC-dimension at most $2\log_2(6)(q+2^{q+1}-1)$.
    Thus, $\mathcal{S}$ has VC-dimension at most $2\log_2(6)(q+2^{q+1}-1)$, and hence so does $\widetilde{c}$.
\end{proof}

After reducing to symmetric colorings, we use the radial structure supplied by \cref{conj:radial} to construct a large shattered set.

\radialimpliesghs* 

\begin{proof}
Assume that the Near-Spanning Radial Ramsey Conjecture holds.
We prove that for all $0 < \varepsilon < 1/2$ and $q \in \mathbb{N}$, there exists $n \in \mathbb{N}$ such that every $\varepsilon$-gap coloring $c : Q_n \times Q_n \to \{0, 1\}$ has VC-dimension at least $q$.
The same coordinate-repetition argument as in the deduction of \cref{thm:main-cube} from \cref{thm:exists-torus} upgrades this statement to all sufficiently large dimensions, so we omit the details.

Let $0 < \varepsilon < 1/2$ and $q \in \mathbb{N}$ be given.
Choose an integer $s > B(q-1)$, where $B(\cdot)$ comes from \cref{lem:make-c-symmetric}.
Choose $m \in \mathbb{N}$ large enough so that $\frac{s+1}{m} \leq \varepsilon$ and $m-s+1 \geq (1-\varepsilon/2) \cdot m$.
Let $n$ be the integer given by Near-Spanning Radial Ramsey for $\varepsilon/2$ and $m$.
By \cref{lem:make-c-symmetric}, it suffices to prove that every symmetric $\varepsilon$-gap coloring $c : Q_n \times Q_n \to \{0, 1\}$ has VC-dimension at least $s$.
Indeed, if an arbitrary $\varepsilon$-gap coloring had VC-dimension less than $q$, then its VC-dimension would be at most $q-1$.
By \cref{lem:make-c-symmetric}, there would exist a symmetric $\varepsilon$-gap coloring of VC-dimension at most $B(q-1) < s$, a contradiction.

Let $c : Q_n \times Q_n \to \{0, 1\}$ be a symmetric $\varepsilon$-gap coloring.
By Near-Spanning Radial Ramsey, there exists a $c$-radial $\varepsilon/2$-near-spanning homothetic copy $\phi : Q_m \to Q_n$ of $Q_m$.
Let $f : \{0, \ldots, m\} \to \{0, 1\}$ be such that $c(\phi(x), \phi(y)) = f(d(x, y))$ for all $x, y \in Q_m$.
Let $\lambda > 0$ be the ratio of $\phi$.
Since $\phi$ is homothetic, $\diam(\phi(Q_m)) = \lambda m$.
Consequently, since $\phi$ is $\varepsilon/2$-near-spanning, we have $(1-\varepsilon/2)n \leq \lambda m \leq n$, so $\frac{1-\varepsilon/2}{m}n \leq \lambda \leq \frac{1}{m} \cdot n$.

Pick arbitrary $x, y \in Q_m$ such that $d(x, y) \leq s+1$.
Then, \[d(\phi(x), \phi(y)) \leq \lambda \cdot (s+1) \leq \frac{s+1}{m} \cdot n \leq \varepsilon \cdot n.\]
Since $c$ is an $\varepsilon$-gap coloring, this implies $c(\phi(x), \phi(y)) = 0$. 
Thus, $f(i) = 0$ for every $i \leq s+1$.

Pick arbitrary $x, y \in Q_m$ such that $d(x, y) = m-s+1$.
Then, \[d(\phi(x), \phi(y)) = \lambda \cdot (m-s+1) \geq \frac{1-\varepsilon/2}{m} \cdot n \cdot (1 - \varepsilon/2) \cdot m \geq (1 - \varepsilon) \cdot n.\]
Since $c$ is an $\varepsilon$-gap coloring, this implies $c(\phi(x), \phi(y)) = 1$. 
Thus, $f(m-s+1) = 1$.

Let $\tau \in \{0, \ldots, m\}$ be minimal such that $f(\tau) = 1$; in particular $\tau\geq s+2$ and $\tau \leq m-s+1$.
We show that $\phi(e_1),\ldots,\phi(e_{s})$ is shattered, where $e_i = \mathbbm{1}_{\{i\}}$ is the $i$-th unit vector.
For this, we construct, for each $J\subseteq[s]$, a witness at distance $\tau$ from $e_i$ exactly when $i\in J$, and at distance $\tau-2$ otherwise.
Let $J\subseteq[s]$.
Let $S_J = ([s] \setminus J) \cup \{s + i : i \in J\} \cup \llbracket2s+1, s+\tau-1\rrbracket$ and let $x_J = \mathbbm{1}_{S_J}$.
The set $[s]\setminus J$ ensures that $x_J$ contains precisely the basis coordinates indexed outside $J$, while the shifted coordinates $s+i$ compensate for those removed coordinates and hence keep the weight unchanged; the final interval then pads the support so that $|S_J|=\tau-1$.
Note that this is possible and that $|S_J| = \tau-1$ since $\tau\geq s+2$ and $\tau \leq m-s+1$, so $2s+1 \leq s+\tau-1 \leq m$.
Thus, $d(x_J, e_i) = \tau$ for every $i \in J$ and $d(x_J, e_i) = \tau-2$ for every $i \notin J$.
Therefore, for every $i \in J$, we have $c(\phi(x_J), \phi(e_i)) = f(\tau) = 1$, and for every $i \notin J$, we have $c(\phi(x_J), \phi(e_i)) = f(\tau-2) = 0$ by minimality of $\tau$.
This proves that $c$ has VC-dimension at least $s$, which concludes the proof.
\end{proof}

\subsection{From homothetic copies to balanced partitions}

This section gives some intuition on the relation between homothetic maps and balanced partitions. It will not be used in the rest of the proof but helps to explain the link between \cref{conj:radial} and \cref{conj:hdr}.

Given two integers $m \leq n$, we describe the homothetic maps $\phi : Q_m \to Q_n$.
Informally, the next result says that, up to an automorphism of $Q_n$, every homothetic map is just $x \mapsto (\underbrace{x_1,\ldots,x_1}_{\lambda\text{ times}},
    \ldots,
    \underbrace{x_m,\ldots,x_m}_{\lambda\text{ times}},
    \underbrace{0,\ldots,0}_{n-\lambda m\text{ times}}).$

\begin{lemma}\label{lem:form-homothetic-map}
    Let $m \leq n$ and let $\phi : Q_m \to Q_n$ be a map.
    Then, $\phi$ is homothetic with ratio $\lambda > 0$ if and only if there exist pairwise disjoint sets $U_0, \ldots, U_m$ with union $[n]$ such that $|U_1| = \ldots = |U_m| = \lambda$ and a vector $z \in Q_n$ such that for every $x \in Q_m$ and every coordinate $t \in [n]$, we have $\phi(x)_t = z_t$ if $t \in U_0$ and $\phi(x)_t = z_t \oplus x_i$ if $t \in U_i$ for $i \geq 1$.
\end{lemma}

\begin{proof}
    Suppose first that $\phi$ is homothetic with ratio $\lambda > 0$.
    Let $e_1, \ldots, e_m \in Q_m$ be the unit vectors.
    Set $z = \phi(0^m)$, for every $i \in [m]$ let $U_i$ be the set of coordinates on which $\phi(0^m)$ and $\phi(e_i)$ differ and let $U_0 = [n] \setminus (U_1 \cup \ldots \cup U_m)$.
    For every $i \in [m]$, we have $\lambda = \lambda \cdot d(0^m, e_i) = d(\phi(0^m), \phi(e_i)) = |U_i|$.
    Furthermore, for all distinct $i, j \in [m]$, the vectors $\phi(e_i)$ and $\phi(e_j)$ can only differ in coordinates in $U_i \cup U_j$. Thus, $|U_i \cup U_j| \geq d(\phi(e_i), \phi(e_j)) = \lambda \cdot d(e_i, e_j) = 2\lambda$.
    However, $|U_i \cup U_j| \leq |U_i| + |U_j| = 2\lambda$, so $|U_i \cup U_j| = 2\lambda$, which implies that $U_i$ and $U_j$ are disjoint.
    Thus, $U_0, U_1, \ldots, U_m$ form a partition of $[n]$.
    
    Now, let $x \in Q_m$ and let $S = \{i \in [m] : x_i = 1\}$.
    Note that $d(\phi(0^m), \phi(x)) = \lambda |S|$ and for every $i \in S$ we have $d(\phi(e_i), \phi(x)) = \lambda (|S|-1)$.
    However, the vectors $\phi(0^m)$ and $\phi(e_i)$ differ exactly on the $\lambda$ coordinates in $U_i$.
    Thus, the vectors $\phi(x)$ and $\phi(0^m)$ differ on all the coordinates in $U_i$.
    Since $|\bigcup_{i \in S}U_i| = \lambda |S| = d(\phi(0^m), \phi(x))$, the vectors $\phi(0^m)$ and $\phi(x)$ differ exactly on the coordinates in $\bigcup_{i \in S}U_i$.
    This proves that for every coordinate $t \in [n]$, we have $\phi(x)_t = z_t$ if $t \in U_0$ and $\phi(x)_t = z_t \oplus x_i$ if $t \in U_i$ for $i \geq 1$.
    
    Conversely, suppose that there exist a partition of $[n]$ into sets $U_0, \ldots, U_m$ with $|U_1| = \ldots = |U_m| = \lambda$ and a vector $z \in Q_n$ such that for every $x \in Q_m$ and every coordinate $t \in [n]$, we have $\phi(x)_t = z_t$ if $t \in U_0$ and $\phi(x)_t = z_t \oplus x_i$ if $t \in U_i$ for $i \geq 1$.
    Then, for all $x, y \in Q_m$, the vectors $\phi(x)$ and $\phi(y)$ disagree exactly on the coordinates in $\bigcup_{i : x_i \neq y_i}U_i$.
    Thus, $d(\phi(x), \phi(y)) = |\bigcup_{i : x_i \neq y_i}U_i| = \lambda \cdot d(x, y)$, as desired.
\end{proof}

The classification is included for context; the subsequent construction uses only the balanced block maps described below.
We take $|U_0|=|U_1|=\cdots=|U_m|=\lambda$, with $\lambda$ even, and choose $z$ so that it has $\lambda/2$ zeros and $\lambda/2$ ones inside every $U_i$. 
Splitting each $U_i$ according to the value of $z$ produces $2m+2$ equal-sized blocks.

\subsection{Yet another conjecture}

Pairs of cube vertices naturally determine partitions of the coordinates into four classes with unequal class sizes. 
We therefore need a version of Homogeneous Dual Ramsey in which these proportions are prescribed.
A \emph{profile} is a finite vector of positive rational numbers summing to $1$.
We write $p(\alpha)$ for its number of entries.

Let $\mathcal{A}=(A_1,\ldots,A_s)$ be a partition of $[n]$. 
Number its classes so that $\min(A_1) < \ldots < \min(A_s)$.
We then say that $\mathcal{A}$ is \emph{canonically ordered}.
Equivalently, $\mathcal{A}$ may be represented by the surjective function $\pi\colon [n]\to[s]$ defined by $\pi(i)=j$ if and only if $i\in A_j$.
This surjection is \emph{rigid}, meaning that $\min(\pi^{-1}(1)) < \ldots < \min(\pi^{-1}(s))$.
Similarly, let $\mathcal{A}'=(A'_1,\ldots,A'_t)$ be a canonically ordered partition of $[n]$, represented by rigid surjection $\pi'\colon[n]\to[t]$. 
Then $\mathcal{A}$ refines $\mathcal{A}'$ if and only if there exists a rigid surjection $\psi\colon[s]\to[t]$ such that $\pi'=\psi\circ\pi$. 
Indeed, the function $\psi$ records which parts of $\mathcal{A}$ are merged to form each part of $\mathcal{A}'$.
Importantly, the composition of two rigid surjections is still a rigid surjection.
From now on, we identify partitions with their corresponding rigid surjections.

The \emph{profile} of a canonically ordered partition $\mathcal{A}=(A_1,\ldots,A_s)$ is the vector $\left(\frac{|A_1|}{n},\ldots,\frac{|A_s|}{n}\right)$. 
If this vector is equal to $\alpha=(\alpha_1,\ldots,\alpha_s)$, we call $\mathcal{A}$ an \emph{$\alpha$-partition}. 
Thus, the profile of a partition $\mathcal{A}$ with corresponding rigid surjection $\pi$ is $\left(\frac{|\pi^{-1}(1)|}{n},\ldots,\frac{|\pi^{-1}(s)|}{n}\right)$.

Let $\alpha=(\alpha_1,\ldots,\alpha_s)$ and $\beta=(\beta_1,\ldots,\beta_t)$ be two profiles. 
A \emph{copy} of $\beta$ in $\alpha$ is a rigid surjection $\psi\colon[s]\to[t]$ such that for every $j\in[t]$, we have $\sum_{i\in\psi^{-1}(\{j\})}\alpha_i=\beta_j$. 
The rigid surjection $\psi$ describes how to merge the parts of an $\alpha$-partition to obtain a $\beta$-partition.
More precisely, if $\pi\colon[n]\to[s]$ is an $\alpha$-partition and $\psi : [s] \to [t]$ is a copy of $\beta$ in $\alpha$, then $\psi\circ\pi\colon[n]\to[t]$ is a $\beta$-partition. 
Conversely, every coarsening of an $\alpha$-partition with profile $\beta$ arises via a copy of $\beta$ in $\alpha$.
Therefore, there exists a copy of $\beta$ in $\alpha$ if and only if every $\alpha$-partition admits a coarsening that is a $\beta$-partition. In this case, we say that $\alpha$ \emph{refines} $\beta$.

We now state a conjecture that generalizes the Homogeneous Dual Ramsey Conjecture, and we will later show that it actually follows from the Homogeneous Dual Ramsey Conjecture.

\begin{conjecture}[Multi-Profile Dual Ramsey]\label[conjecture]{conj:mpdr}
Let $\alpha^{(1)},\ldots,\alpha^{(s)}$ and $\beta$ be profiles, and suppose that $\beta$ refines $\alpha^{(i)}$ for every $i\in[s]$.
Then, for every $r\in\mathbb{N}$, there exists an integer $N$ with the following property.

For each $i\in[s]$, let $c_i$ be an $r$-coloring of the $\alpha^{(i)}$-partitions of $[N]$.
Then there exists a $\beta$-partition $\pi$ of $[N]$ such that, for every $i\in[s]$, all the $\alpha^{(i)}$-partitions obtained by coarsening $\pi$ have the same color under $c_i$.

Equivalently, for every $i\in[s]$, the color $c_i(\psi\circ\pi)$ is independent of the choice of copy $\psi$ of $\alpha^{(i)}$ in $\beta$. However, this color may depend on the profile index $i$.
\end{conjecture}

For an integer $m$, we denote by $\unif{m}$ the \emph{homogeneous profile} $\left(\frac{1}{m}, \ldots, \frac{1}{m}\right)$. In other words, an $\unif{m}$-partition is an $m$-equipartition.
Observe that the Homogeneous Dual Ramsey Conjecture is simply the Multi-Profile Dual Ramsey Conjecture with $s = 1$ and the profiles $\alpha = \unif{k}$ and $\beta = \unif{m}$.

\subsection{Homogeneous Dual Ramsey implies Multi-Profile Dual Ramsey}\label{subsec:hdr-implies-mpdr}

We first prove an amalgamation lemma, which we then use to deduce the single-profile case of the Multi-Profile Dual Ramsey Conjecture from the Homogeneous Dual Ramsey Conjecture. It is then straightforward to establish the Multi-Profile Dual Ramsey Conjecture from its single-profile case.

Our argument gives an explicit rational-profile formulation of the method of Kechris, Sokić and Todorčević \cite{KST17}. 
The Amalgamation Lemma (\cref{lem:common-refinement}) follows the construction in the proof of their Theorem 3.1. The deduction of Profile Dual Ramsey (\cref{lem:hdr-implies-pdr}) follows their cofinal-subclass transfer, Proposition 5.3, including the iterative amalgamation in Lemma 5.4.

\begin{lemma}[Amalgamation Lemma]\label[lemma]{lem:common-refinement}
Let $\alpha, \beta, \gamma$ be profiles such that $\beta$ and $\gamma$ both refine $\alpha$.
Let $\phi_\alpha$ be a copy of $\alpha$ in $\beta$ and $\psi_\alpha$ be a copy of $\alpha$ in $\gamma$.
Then, there exists a profile $\delta$ which refines both $\beta$ and $\gamma$, a copy $\phi_\beta$ of $\beta$ in $\delta$ and a copy $\psi_{\gamma}$ of $\gamma$ in $\delta$ such that $\phi_{\alpha} \circ \phi_{\beta} = \psi_{\alpha} \circ \psi_{\gamma}$. See \cref{fig:commutative-diagram} for an illustration of the statement.
\end{lemma}

\begin{figure}[ht]
    \centering
    \includegraphics[height=12\baselineskip]{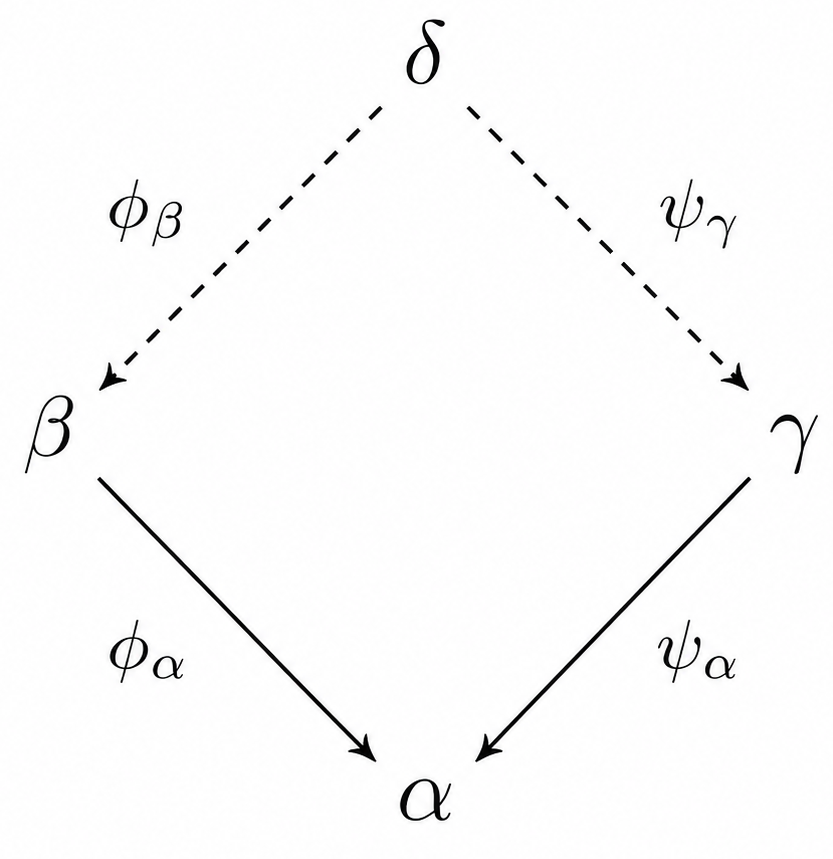}
    \caption{An illustration of the statement of \cref{lem:common-refinement}. The solid arrows are the given copies, the dashed arrows are constructed, and the diagram commutes.}
    \label{fig:commutative-diagram}
\end{figure}

\begin{proof}
Write $\alpha = (a_1, \ldots, a_{p(\alpha)})$, $\beta = (b_1, \ldots, b_{p(\beta)})$ and $\gamma = (c_1, \ldots, c_{p(\gamma)})$.
Thus, we have $\phi_{\alpha} : [p(\beta)] \to [p(\alpha)]$ and $\psi_{\alpha} : [p(\gamma)] \to [p(\alpha)]$.
We first distribute mass independently within each coarse part. 
We then order the resulting pieces so that the two projections preserve first occurrences.

For every $i \in [p(\alpha)]$, we have $\sum_{j \in \phi_{\alpha}^{-1}(i)}b_j=a_i$ and $\sum_{k\in \psi_{\alpha}^{-1}(i)}c_k=a_i$.
For every $i \in [p(\alpha)]$, let $S_i = \{(j, k) \in [p(\beta)] \times [p(\gamma)] : \phi_{\alpha}(j) = i = \psi_{\alpha}(k)\}$. Let $S = \bigcup_{i \in [p(\alpha)]} S_i$.
For every $i \in [p(\alpha)]$ and all $(j, k) \in S_i$, set
$\delta_{j,k}=\frac{b_j c_k}{a_i} > 0$.
Note that \begin{align*}
    \sum_{(j, k) \in S}\delta_{j,k} &= \sum_{i \in [p(\alpha)]}\sum_{(j, k) \in S_i}\delta_{j,k} = \sum_{i \in [p(\alpha)]} \sum_{j \in \phi_{\alpha}^{-1}(i)}\sum_{k \in \psi_{\alpha}^{-1}(i)}\frac{b_j c_k}{a_i} \\
    &= \sum_{i \in [p(\alpha)]} \frac{1}{a_i}\left(\sum_{j \in \phi_{\alpha}^{-1}(i)}b_j\right)\left(\sum_{k \in \psi_{\alpha}^{-1}(i)} c_k\right) = \sum_{i \in [p(\alpha)]}\frac{a_i^2}{a_i} =  \sum_{i \in [p(\alpha)]}a_i = 1.
\end{align*}

For every $i \in [p(\alpha)]$, let $u_i = \min(\phi_{\alpha}^{-1}(i)) \in [p(\beta)]$ and $v_i = \min(\psi_{\alpha}^{-1}(i)) \in [p(\gamma)]$.
Since $\phi_{\alpha}$ and $\psi_{\alpha}$ are rigid surjections, we have $u_1 < \ldots < u_{p(\alpha)}$ and $v_1 < \ldots < v_{p(\alpha)}$.

Let $S_{\beta} = \{(j, v_{\phi_{\alpha}(j)}) : j \in [p(\beta)]\}$ and $S_{\gamma} = \{(u_{\psi_{\alpha}(k)}, k) : k \in [p(\gamma)]\}$.
Note that $S_{\beta}, S_{\gamma} \subseteq S$.
Order $S_{\beta}$ as $(1, v_{\phi_{\alpha}(1)}) \prec_{\beta} (2, v_{\phi_{\alpha}(2)}) \prec_{\beta} \ldots \prec_{\beta} (p(\beta), v_{\phi_{\alpha}(p(\beta))})$ and order $S_{\gamma}$ as $( u_{\psi_{\alpha}(1)}, 1) \prec_{\gamma} (u_{\psi_{\alpha}(2)}, 2) \prec_{\gamma} \ldots \prec_{\gamma} (u_{\psi_{\alpha}(p(\gamma))}, p(\gamma))$.
Note that $S_{\beta}$ and $S_{\gamma}$ intersect exactly on the elements $(u_i, v_i)$ for $i \in [p(\alpha)]$, and these elements are ordered by increasing $i$ both in $\prec_{\beta}$ and $\prec_{\gamma}$.
Thus, there exists a linear order $\prec$ on $S_{\beta} \cup S_{\gamma}$ that extends $\prec_{\beta}$ and $\prec_{\gamma}$.
Extend $\prec$ to $S$ by adding the elements of $S \setminus (S_{\beta} \cup S_{\gamma})$ at the end arbitrarily.
Enumerate the elements of $S$ as $s_1 \prec \ldots \prec s_{p(\delta)}$.
For every $\ell \in [p(\delta)]$, if $s_\ell = (j, k)$, let $d_\ell = \delta_{j, k}$.
Let $\delta$ be the profile $(d_1, \ldots, d_{p(\delta)})$.
For every $\ell \in [p(\delta)]$, if $s_\ell = (j, k)$, let $\phi_{\beta}(\ell) = j$ and $\psi_{\gamma}(\ell) = k$.

We show that $\phi_{\beta}$ is a copy of $\beta$ in $\delta$; the proof that $\psi_{\gamma}$ is a copy of $\gamma$ in $\delta$ is similar.
First, let $j \in [p(\beta)]$ and let $i = \phi_{\alpha}(j)$. 
Then, \[\sum_{\ell \in \phi_{\beta}^{-1}(j)}d_{\ell} = \sum_{k \in \psi_{\alpha}^{-1}(i)}\delta_{j, k} = \sum_{k \in \psi_{\alpha}^{-1}(i)}\frac{b_j c_k}{a_i} = \frac{b_j}{a_i} \cdot \sum_{k \in \psi_{\alpha}^{-1}(i)}c_k = b_j.\]

We now prove that $\phi_{\beta}$ is rigid.
For each $j \in [p(\beta)]$, let $\ell_j \in [p(\delta)]$ be the only index such that $s_{\ell_j} = (j, v_{\phi_{\alpha}(j)})$, and note that $s_{\ell_j} \in S_{\beta}$.
We first argue that $\ell_j = \min(\phi_{\beta}^{-1}(j))$.
Indeed, let $\ell \in \phi_{\beta}^{-1}(j)$ and write $s_{\ell} = (j, k)$.
If $s_{\ell} \notin S_{\beta} \cup S_{\gamma}$ then $s_{\ell_j} \prec s_{\ell}$ by definition of $\prec$ on $S$, so $\ell_j \leq \ell$.
If $s_{\ell} \in S_{\beta}$ then $s_{\ell} = (j, v_{\phi_{\alpha}(j)}) = s_{\ell_j}$ so $\ell = \ell_j$.
Finally, suppose that $s_{\ell} \in S_{\gamma}$ and let $i = \phi_{\alpha}(j) = \psi_{\alpha}(k)$.
Since $s_{\ell} \in S_{\gamma}$, we have $j = u_i$.
Moreover, $v_i \leq k$ so $s_{\ell_j} = (j, v_i) = (u_i, v_i) \preceq_{\gamma} (u_i, k) = (j, k) = s_{\ell}$.
Therefore, $s_{\ell_j} \preceq s_{\ell}$ so $\ell_j \leq \ell$.
This concludes the proof that $\ell_j = \min(\phi_{\beta}^{-1}(j))$.

Now, if $j, j' \in [p(\beta)]$ and $j < j'$ then $s_{\ell_j} = (j, v_{\phi_{\alpha}(j)}) \prec_{\beta} (j', v_{\phi_{\alpha}(j')}) = s_{\ell_{j'}}$ so $s_{\ell_j} \prec s_{\ell_{j'}}$, and thus $\ell_j < \ell_{j'}$.
This concludes the proof that $\phi_{\beta}$ is rigid.
Moreover, for every $j \in [p(\beta)]$, we have $\phi_{\beta}(\ell_j) = j$, so $\phi_{\beta} : [p(\delta)] \to [p(\beta)]$ is a rigid surjection, hence a copy of $\beta$ in $\delta$.

Finally, we argue that $\phi_{\alpha} \circ \phi_{\beta} = \psi_{\alpha} \circ \psi_{\gamma}$.
For this, consider an arbitrary $\ell \in [p(\delta)]$, and let $(j, k) = s_\ell \in S$. Let $i \in [p(\alpha)]$ such that $(j, k) \in S_i$, so $\phi_{\alpha}(j) = i = \psi_{\alpha}(k)$.
Then, by definition of $\phi_{\beta}$ and $\psi_{\gamma}$, we have \[\phi_{\alpha} \circ \phi_{\beta}(\ell) = \phi_{\alpha}(j) = i = \psi_{\alpha}(k) = \psi_{\alpha} \circ \psi_{\gamma} (\ell).\]
\end{proof}

\begin{observation}\label[observation]{obs:canonical-coarsening}
Let $\alpha=(\alpha_1,\ldots,\alpha_k)$ be a profile. 
There exists $N\in\mathbb N$ such that $\unif{N}$ refines $\alpha$. 
Indeed, one may take any $N$ for which $N\alpha_i\in\mathbb N$ for every $i\in[k]$.

There is a canonical copy of $\alpha$ in $\unif{N}$, obtained by merging consecutive parts into blocks of sizes $N\alpha_1,\ldots,N\alpha_k$. 
We denote this copy by $\kappa : [N] \to [k]$. 
\end{observation}

We can now prove that the single-profile case of the Multi-Profile Dual Ramsey Conjecture follows from the Homogeneous Dual Ramsey Conjecture.

\begin{conjecture}[Profile Dual Ramsey] \label{conj:pdr}
Let $\alpha$ and $\beta$ be profiles such that $\beta$ refines $\alpha$.
Then, for every $r\in\mathbb{N}$, there exists an integer $N$ with the following property.

Let $c$ be an $r$-coloring of the $\alpha$-partitions of $[N]$.
Then there exists a $\beta$-partition $\pi$ of $[N]$ such that all the $\alpha$-partitions obtained by coarsening $\pi$ have the same color under $c$.

Equivalently, the color $c(\psi\circ\pi)$ is independent of the choice of copy $\psi$ of $\alpha$ in $\beta$.
\end{conjecture}

\begin{lemma}\label{lem:hdr-implies-pdr}
    The Homogeneous Dual Ramsey Conjecture implies the Profile Dual Ramsey Conjecture.
\end{lemma}

\begin{proof}
We construct a refinement in which every $\alpha$-coarsening of the distinguished copy of $\beta$ factors through a fixed canonical coarsening of a homogeneous profile.  
We then apply Homogeneous Dual Ramsey to canonicalize all such $\unif{D}$-copies simultaneously.

By \cref{obs:canonical-coarsening}, there exists $D \in \mathbb{N}$ such that $\unif{D}$ refines $\alpha$.
Let $\kappa$ denote the canonical copy of $\alpha$ in $\unif{D}$ given by coarsening.

There are only finitely many copies of $\alpha$ in $\beta$; enumerate them as $\psi_1,\ldots,\psi_t$. 
We first construct a profile $\gamma$, a copy $\phi$ of $\beta$ in $\gamma$, and copies $\chi_1,\ldots,\chi_t$ of $\unif{D}$ in $\gamma$ such that \[\psi_i\circ\phi=\kappa\circ\chi_i, \quad \text{for every $i\in[t]$}.\] 
In other words, every copy of $\alpha$ in the distinguished copy $\phi$ of $\beta$ in $\gamma$ is the canonical copy of $\alpha$ in some copy of $\unif{D}$ in $\gamma$.
We impose these identities one at a time by amalgamation. 
At each step, composing all previously constructed maps with the new refinement preserves the identities already obtained.

Set $\gamma_0=\beta$ and let $\phi_0$ be the identity copy of $\beta$ in $\gamma_0$. 
Suppose that, for some $j\in[t]$, we have constructed a profile $\gamma_{j-1}$, a copy $\phi_{j-1}$ of $\beta$ in $\gamma_{j-1}$, and for each $i<j$, a copy $\chi_i^{j-1}$ of $\unif{D}$ in $\gamma_{j-1}$ such that $\psi_i\circ\phi_{j-1}=\kappa\circ\chi_i^{j-1}$.

Then, $\psi_j\circ\phi_{j-1}$ is a copy of $\alpha$ in $\gamma_{j-1}$, while $\kappa$ is a copy of $\alpha$ in $\unif{D}$. 
By the Amalgamation Lemma (\cref{lem:common-refinement}), there exist a profile $\gamma_j$, a copy $\theta_j$ of $\gamma_{j-1}$ in $\gamma_j$, and a copy $\chi_j^j$ of $\unif{D}$ in $\gamma_j$ such that $\psi_j\circ\phi_{j-1}\circ\theta_j=\kappa\circ\chi_j^j$.

Set $\phi_j=\phi_{j-1}\circ\theta_j$. 
For each $i<j$, set $\chi_i^j=\chi_i^{j-1}\circ\theta_j$. 
Then $\phi_j$ is a copy of $\beta$ in $\gamma_j$, and, for every $i<j$, $\chi_i^j$ is a copy of $\unif{D}$ in $\gamma_j$ such that $\psi_i\circ\phi_j=\psi_i\circ\phi_{j-1}\circ\theta_j=\kappa\circ\chi_i^{j-1}\circ\theta_j=\kappa\circ\chi_i^j$. 
The corresponding identity for $i=j$ follows from the choice of $\chi_j^j$. 
This completes the induction.

Taking $\gamma=\gamma_t$, $\phi=\phi_t$, and $\chi_i=\chi_i^t$ for every $i\in[t]$, we have that $\phi$ is a copy of $\beta$ in $\gamma$, and $\chi_i$ is a copy of $\unif{D}$ in $\gamma$ such that $\psi_i\circ\phi=\kappa\circ\chi_i$ for every $i\in[t]$, as desired.

Choose $M$ divisible by $D$ and by the denominators of all entries of $\gamma$, so that $\unif{M}$ refines both $\gamma$ and $\unif{D}$, and fix a copy $\rho$ of $\gamma$ in $\unif{M}$.
Since $\unif{M}$ refines $\unif{D}$, we have $D \mid M$.
Fix the number of colors $r$, and let $N$ be the integer given by Homogeneous Dual Ramsey Conjecture for $D$, $M$, and $r$.
Then, $N$ is divisible by $M$, and every $r$-coloring of the $\unif{D}$-partitions of $[N]$ admits an $\unif{M}$-partition $\xi\colon[N]\to[M]$ for which all copies of $\unif{D}$ that coarsen $\xi$ have the same color.

Let $c$ be an $r$-coloring of the $\alpha$-partitions of $[N]$. 
Define an $r$-coloring $c^*$ of the $\unif{D}$-partitions of $[N]$ by $c^*(\theta)=c(\kappa \circ \theta)$.

By assumption, there exists an $\unif{M}$-partition $\xi\colon[N]\to[M]$ such that all copies of $\unif{D}$ that coarsen $\xi$ have the same color under $c^*$.
Consider the $\beta$-partition $\pi=\phi\circ\rho\circ\xi$ of $[N]$.

Every $\alpha$-partition that coarsens $\pi$ is of the form $\psi_i \circ \pi = \psi_i\circ\phi\circ\rho\circ\xi$ for some $i\in[t]$. 
By the construction of $\gamma$, we have $\psi_i\circ\phi=\kappa\circ\chi_i$. 
Observe that $\chi_i\circ\rho\circ\xi$ is an $\unif{D}$-partition of $[N]$ coarsening $\xi$, and
$c(\psi_i\circ\pi)=c(\psi_i\circ\phi\circ\rho\circ\xi)=c(\kappa\circ\chi_i\circ\rho\circ\xi)=c^\star(\chi_i\circ\rho\circ\xi)$.

Since all copies of $\unif{D}$ that coarsen $\xi$ have the same color under $c^*$, the right-hand side is independent of $i$. It follows that all $\alpha$-partitions coarsening $\pi$ have the same color under $c$. Thus, $\pi$ is a $\beta$-partition of $[N]$ with the required property.
\end{proof}

The Multi-Profile Dual Ramsey Conjecture then follows quickly from the Profile Dual Ramsey Conjecture by induction on the number of profiles.

\begin{lemma}\label{lem:pdr-implies-mpdr}
    The Profile Dual Ramsey Conjecture implies the Multi-Profile Dual Ramsey Conjecture.
\end{lemma}

\begin{proof}
    We prove it by induction on the number $k$ of profiles.
    The case $k=1$ is exactly the Profile Dual Ramsey Conjecture.
    Assume that we have established the Multi-Profile Dual Ramsey Conjecture for $k$ profiles and let $\alpha^{(1)},\ldots,\alpha^{(k+1)}$ and $\beta$ be given so that $\beta$ refines all $\alpha^{(i)}$.
    Fix the desired number of colors $r$, and let $M \in \mathbb{N}$ be given by the induction hypothesis for the profiles $\alpha^{(2)},\ldots,\alpha^{(k+1)}$ and $\beta$ with $r$ colors.
    Since there exists a $\beta$-partition of $[M]$, it follows that $\unif{M}$ refines $\beta$ and therefore refines all $\alpha^{(i)}$.
    Let $N \in \mathbb{N}$ be given by the single-profile case for the profiles $\alpha^{(1)}$ and $\unif{M}$ with $r$ colors.
    
    For each $i \in [k+1]$, let $c_i$ be an $r$-coloring of the $\alpha^{(i)}$-partitions of $[N]$.
    By the single-profile case for the profiles $\alpha^{(1)}$ and $\unif{M}$ with $r$ colors, there exists an $\unif{M}$-partition $\xi$ of $[N]$ such that all the $\alpha^{(1)}$-partitions obtained by coarsening $\xi$ have the same color under $c_1$.
    Let $2 \leq i \leq k+1$ and let $\psi$ be an $\alpha^{(i)}$-partition of $[M]$.
    Then, $\psi \circ \xi$ is an $\alpha^{(i)}$-partition of $[N]$.
    Set $d_i(\psi) = c_i(\psi \circ \xi)$, so $d_i$ is an $r$-coloring of the $\alpha^{(i)}$-partitions of $[M]$.
    By the induction hypothesis for the profiles $\alpha^{(2)},\ldots,\alpha^{(k+1)}$ and $\beta$ with $r$ colors, there exists a $\beta$-partition $\mu$ of $[M]$ such that, for every $2 \leq i \leq k+1$, all the $\alpha^{(i)}$-partitions of $[M]$ obtained by coarsening $\mu$ have the same color under $d_i$.
    Then, set $\pi = \mu \circ \xi$, so $\pi$ is a $\beta$-partition of $[N]$.
    Moreover, every $\alpha^{(1)}$-partition that coarsens $\pi$ also coarsens $\xi$, so they all have the same color under $c_1$.
    Finally, for $2 \leq i \leq k+1$, every $\alpha^{(i)}$-partition that coarsens $\pi$ is of the form $\phi \circ \pi = \phi \circ \mu \circ \xi$. 
    Then, $\phi \circ \mu$ is an $\alpha^{(i)}$-partition of $[M]$ obtained by coarsening $\mu$ and $d_i(\phi \circ \mu) = c_i(\phi \circ \mu \circ \xi) = c_i(\phi \circ \pi)$.
    The left-hand side is constant by assumption, which proves that all the $\alpha^{(i)}$-partitions that coarsen $\pi$ have the same color under $c_i$.
\end{proof}

Combining \cref{lem:hdr-implies-pdr,lem:pdr-implies-mpdr}, we obtain the following.

\begin{corollary}\label{cor:hdr-implies-mpdr}
    The Homogeneous Dual Ramsey Conjecture implies the Multi-Profile Dual Ramsey Conjecture.
\end{corollary}

Since the Multi-Profile Dual Ramsey Conjecture generalizes the Homogeneous Dual Ramsey Conjecture, these two conjectures are in fact equivalent.

\subsection{Multi-Profile Dual Ramsey implies Near-Spanning Radial Ramsey}

In this section, we complete the proof of \cref{thm:hdr-implies-radial}.
We first prove the following statement.

\begin{lemma}\label{lem:mpdr-implies-somewhat-canonical}
Assuming the Multi-Profile Dual Ramsey Conjecture, the following statement holds.
For every $m\in\N$, there exists $n\in\N$ such that for every symmetric coloring $c : Q_n \times Q_n \to \{0, 1\}$, there exists a $c$-radial $\frac{1}{m+1}$-near-spanning homothetic copy of $Q_m$.
\end{lemma}

\begin{proof}
Let $\alpha^{(0)} = \unif{2} = (1/2, 1/2)$ and for $i\in [m]$, let
\[
  \alpha^{(i)}
  =\left(
      \frac{m+1-i}{2m+2},
      \frac{m+1-i}{2m+2},
      \frac{i}{2m+2},
      \frac{i}{2m+2}
    \right),
\]
and let
\[
  \beta
  =\left(
      \underbrace{\frac1{2m+2},\ldots,\frac1{2m+2}}_{2m+2\text{ entries}}
    \right).
\]
The profile $\beta$ refines $\alpha^{(i)}$ for every $i$.
Let $n$ be the integer given by Multi-Profile Dual Ramsey for the profiles $\alpha^{(0)}, \ldots, \alpha^{(m)}$ and $\beta$ for $2$ colors.

Let $c : Q_n \times Q_n \to \{0, 1\}$ be a symmetric coloring.
Let $(P_1, P_0)$ be a canonically ordered $\alpha^{(0)}$-partition of $[n]$\footnote{The reason for these names will become clearer later in the proof.}.
Define $c_0(P_1, P_0) = c(\mathbbm{1}_{P_1}, \mathbbm{1}_{P_1})$.
Let $j \in [m]$ and let $(P_{11}, P_{00}, P_{10}, P_{01})$ be a canonically ordered $\alpha^{(j)}$-partition of $[n]$.
Define $c_j(P_{11}, P_{00}, P_{10}, P_{01}) = c(\mathbbm{1}_{P_{11} \cup P_{10}}, \mathbbm{1}_{P_{11} \cup P_{01}})$.

By Multi-Profile Dual Ramsey, there exists a $\beta$-partition $\pi$ of $[n]$ such that, for every $j \in \{0, \ldots, m\}$, all the $\alpha^{(j)}$-partitions obtained by coarsening $\pi$ have the same color under $c_j$.
Let $U_0^1, U_0^0, U_1^1, U_1^0, \ldots , U_m^1, U_m^0$ be the parts of $\pi$ in canonical order.

Let \[\phi : Q_m \to Q_n, x \mapsto \mathbbm{1}_{U_0^1 \cup \bigcup_{i \in [m]}U_i^{x_i}}.\]
Thus, each input bit $x_i$ is represented by a pair of equal blocks $U_i^0,U_i^1$, and $\phi(x)$ selects exactly one of them according to the value of $x_i$. 
The two anchor blocks $U_0^1$ and $U_0^0$ ensure that the common-one and common-zero parts occur first and second in canonical order.
For $x, y \in Q_m$ at Hamming distance $j$, exactly $j$ coordinate pairs contribute to each disagreement part, while the remaining $m-j$ coordinate pairs and the two anchors contribute to the agreement parts. This is precisely the profile $\alpha^{(j)}$.

Since all sets $U_i^b$ have size $\frac{n}{2m+2}$, it can be easily checked that $\phi$ is homothetic with ratio $\frac{n}{m+1}$.
Then, $\diam(\phi(Q_m)) = \frac{n}{m+1} \cdot \diam(Q_m) = \frac{m}{m+1} \cdot \diam(Q_n)$, so $\phi$ is $\frac{1}{m+1}$-near-spanning.
Let $x \in Q_m$ and let $P_1 = \{i \in [n] : \phi(x)_i = 1\}$ and $P_0 = [n] \setminus P_1$.
Since $U^1_0 \subseteq P_1$, the partition $(P_{1}, P_0)$ is canonically ordered.
It follows from the definition of $\phi$ that $(P_1, P_0)$ is an $\alpha^{(0)}$-partition of $[n]$.
Then, $c(\phi(x), \phi(x)) = c(\mathbbm{1}_{P_1}, \mathbbm{1}_{P_1}) = c_0(P_1, P_0)$.
Since $(P_1, P_0)$ is an $\alpha^{(0)}$-partition that coarsens $\pi$, the right-hand side is independent of $x$.
Similarly, let $j \in [m]$ and let $x, y \in Q_m$ such that $d(x, y) = j$.
Let $P_{11} = \{i \in [n] : \phi(x)_i = 1, \phi(y)_i = 1\}$, $P_{00} = \{i \in [n] : \phi(x)_i = 0, \phi(y)_i = 0\}$, $P_{10} = \{i \in [n] : \phi(x)_i = 1, \phi(y)_i = 0\}$ and $P_{01} = \{i \in [n] : \phi(x)_i = 0, \phi(y)_i = 1\}$.
The definition of $\phi$ implies that this partition coarsens $\pi$.
Since $U_0^1\subseteq P_{11}$ and $U_0^0\subseteq P_{00}$, the first two coarse parts in the canonical order are $P_{11}$ and $P_{00}$. 
The remaining two parts are $P_{10}$ and $P_{01}$; by exchanging $x$ and $y$ if necessary, we may assume that they occur in this order, since this exchange swaps $P_{10}$ and $P_{01}$ and, by symmetry of $c$, does not change the color.
Then, $c(\phi(x), \phi(y)) = c(\mathbbm{1}_{P_{11} \cup P_{10}}, \mathbbm{1}_{P_{11} \cup P_{01}}) = c_j(P_{11}, P_{00}, P_{10}, P_{01})$.
Since $(P_{11}, P_{00}, P_{10}, P_{01})$ is an $\alpha^{(j)}$-partition that coarsens $\pi$, the right-hand side only depends on $j$.
This shows that $\phi$ is $c$-radial.
\end{proof}

We then easily deduce the following result.

\begin{lemma}\label{lem:mpdr-implies-radial}
    The Multi-Profile Dual Ramsey Conjecture implies Near-Spanning Radial Ramsey Conjecture.
\end{lemma}

\begin{proof}
    Let $0 < \varepsilon < 1/2$ and $m \in \mathbb{N}$ be given.
    Let $t \in \mathbb{N}$ be large enough so that $\frac{1}{tm+1} \leq \varepsilon$.
    Assuming the Multi-Profile Dual Ramsey Conjecture, let $n \in \mathbb{N}$ be given by \cref{lem:mpdr-implies-somewhat-canonical} for $mt$.
    
    Let $c : Q_n \times Q_n \to \{0, 1\}$ be a symmetric coloring.
    By \cref{lem:mpdr-implies-somewhat-canonical}, there exists a $c$-radial $\frac{1}{mt+1}$-near-spanning homothetic copy $\psi : Q_{mt} \to Q_n$ of $Q_{mt}$.
    Since $\frac{1}{tm+1} \leq \varepsilon$, this copy is $\varepsilon$-near-spanning.
    Let \[\rho : Q_m \to Q_{mt}, x \mapsto (\underbrace{x_1,\ldots,x_1}_{t\text{ times}},
    \ldots,
    \underbrace{x_m,\ldots,x_m}_{t\text{ times}}).\]
    It is then straightforward to verify that $\phi = \psi \circ \rho$ is a $c$-radial $\varepsilon$-near-spanning homothetic copy of $Q_m$, as desired.
\end{proof}

\cref{thm:hdr-implies-radial} then follows immediately from \cref{cor:hdr-implies-mpdr,lem:mpdr-implies-radial}.

\end{document}